\documentclass[a4paper,UKenglish]{lipics-v2018}

\usepackage{microtype}
\usepackage{algorithm}
\usepackage[noend]{algorithmic}
\usepackage{xspace}
\usepackage{wrapfig}

\newcommand{\E}{\ensuremath{\mathsf{E}}}

\newcommand{\eps}{\varepsilon}

\newcommand{\R}{\mathbb{R}}

\newcommand{\omt}[1]{}
\newcommand{\etal}{\emph{et al.}\xspace}

\renewcommand{\c}[1]{\ensuremath{\mathcal{#1}}}

\newcommand{\pchan}{\textsf{ComputePartition-Chan}\xspace}
\newcommand{\pchanF}{\textsf{ComputePartition-Chan-Simple}\xspace}
\newcommand{\pmat}{\textsf{ComputePartition-Mat}\xspace}

\newcommand{\Akd}{\textsf{Ham Tree Sample}\xspace}
\newcommand{\Akdd}{\textsf{Double Ham Tree}\xspace}
\newcommand{\AMatD}{\textsf{Mat Poly Dual}\xspace}
\newcommand{\AMatL}{\textsf{Mat Poly Lines}\xspace}
\newcommand{\AMatP}{\textsf{Mat Poly Points}\xspace}
\newcommand{\AChan}{\textsf{Chan}\xspace}
\newcommand{\AChanS}{\textsf{Chan Simple}\xspace}
\newcommand{\ARS}{\textsf{Random Sample}\xspace}
\newcommand{\AEASE}{\textsf{Biased-L2}\xspace}

\newcommand{\branch}{\textsf{Branching Factor}\xspace}
\newcommand{\inS}{\textsf{Input Size}\xspace}
\newcommand{\out}{\textsf{Output Size}\xspace}
\newcommand{\error}{\textsf{Error}\xspace}

\newcommand{\gentestset}{\textsf{BuildTestSet-x}\xspace}
\newcommand{\dtestset}{\textsf{BuildTestSet-Dual}\xspace}
\newcommand{\ptestset}{\textsf{BuildTestSet-Points}\xspace}
\newcommand{\ltestset}{\textsf{BuildTestSet-Lines}\xspace}

\newcommand{\trap}{\textsf{Trapezoid}\xspace}
\newcommand{\polyT}{\textsf{PolyTree}\xspace}
\newcommand{\cutting}{\textsf{CreateCutting}\xspace}

\newcommand{\zone}{\textsf{Zone}}
\newcommand{\Split}{\ensuremath{\mathsf{split}}}

\newcommand{\plog}{\ensuremath{\mathrm{polylog}}}

\title{Practical Low-Dimensional Halfspace Range Space Sampling}

\titlerunning{Range Space Sampling}

\author{Michael Matheny}{University of Utah}{}{}{}

\author{Jeff M. Phillips}{University of Utah}{}{}{Thanks to supported by NSF CCF-1350888, IIS-1251019, ACI-1443046, CNS-1514520, and CNS-1564287}

\authorrunning{M. Matheny and J. M. Phillips}

\Copyright{Michael Matheny and Jeff M. Phillips}

\subjclass{Theory of computation $\rightarrow$ Computational geometry}

\keywords{Partitions, Range Spaces, Sampling, Halfspaces}

\category{}

\relatedversion{}

\supplement{}

\funding{}
\graphicspath{{IpeFigures/}{Figures/}{Ipe/}{Partition_figures/}{Partition_figures/100000_input_plots/}{Cutting_figures/Plots/}{Cutting_figures/Cuttings/}{Discrepancy_figures/}}
\acknowledgements{}

\EventEditors{Yossi Azar, Hannah Bast, and Grzegorz Herman}
\EventNoEds{3}
\EventLongTitle{26th Annual European Symposium on Algorithms (ESA 2018)}
\EventShortTitle{ESA 2018}
\EventAcronym{ESA}
\EventYear{2018}
\EventDate{August 20--22, 2018}
\EventLocation{Helsinki, Finland}
\EventLogo{}
\SeriesVolume{112}
\ArticleNo{62}
\nolinenumbers 

\begin{document}

\maketitle

\begin{abstract}
We develop, analyze, implement, and compare new algorithms for creating $\eps$-samples of range spaces defined by halfspaces which have size sub-quadratic in $1/\eps$, and have runtime linear in the input size and near-quadratic in $1/\eps$.  The key to our solution is an efficient construction of partition trees.  Despite not requiring any techniques developed after the early 1990s, apparently such a result was never explicitly described.  We demonstrate that our implementations, including new implementations of several variants of partition trees, do indeed run in time linear in the input, appear to run linear in output size, and observe smaller error for the same size sample compared to the ubiquitous random sample (which requires size quadratic in $1/\eps$).  This result has direct applications in speeding up discrepancy evaluation, approximate range counting, and spatial anomaly detection.  
\end{abstract}

\section{Introduction}

Taming the relationship between a point set $X \subset \R^d$ and its interaction with halfspaces $\c{H}_d$, has long been a focus of computational geometry.  Understanding and controlling this interaction is at the heart of problems in range searching, linear classification, coresets, and spatial anomaly detection.  This pair $(X,\c{H}_d)$ describes a range space, the combinatorial set of all subsets of $X$ defined by $h \cap X$ for any halfspace $h \in \c{H}_d$.  
In this paper we focus on two specific and closely-interrelated (as it turns out) constructions for $(X,\c{H}_d)$: $\eps$-samples and partitions, defined next.  
 
An \emph{$\eps$-sample} $Y \subset X$ of $(X, \c{H}_d)$ is a small point set that approximately preserves density with respect to halfspaces: for all $h \in \c{H}_d$, and error parameter $\eps \in (0,1)$ it bounds 
\[
\error(X,Y) = \max_{h \in \c{H}_d} \left| \frac{|Y \cap h|}{|Y|} - \frac{|X \cap h|}{|X|} \right| \leq \eps.  
\]

It is known that $\eps$-samples of size $\Theta(1/\eps^{2d/(d+1)})$ exist for halfspaces~\cite{Ale90}, and in general this size may be required~\cite{Mat95a}.  For many years (c.f., \cite{Mat99,Cha01}) such proofs were not constructive, as they relied on the ``partial coloring lemma''; until in 2010 when Bansal~\cite{Ban10} introduced a polynomial time construction.  The runtime of the low-discrepancy coloring on $m$ points was later reduced~\cite{LM12} to $O(m^{3(d+1)} \plog(m))$, this within the standard merge-reduce framework~\cite{CM96} results in a $O(n (1/\eps)^{2d(3d+2)/(d+1)} \plog(1/\eps))$ runtime for sample construction-- which is still not very efficient.  For instance for $d=2$, this requires $O(n (1/\eps)^{10 + 2/3}\plog(1/\eps))$ time.  
A random sample, which can be generated in $O(n + 1/\eps^2)$ time, is an $\eps$-sample of size $O(\frac{1}{\eps^2}(d+\log \frac{1}{\delta}))$ with probability at least $1-\delta$~\cite{VC71,LLS01}. 
 The above discrepancy-based algorithm can be run on the output of this sample to get optimal size, but it only reduces the overall runtime of the $\eps$-sample construction to $O(n + (1/\eps)^{2d(3d+2)/(d+1) + 2} \plog(1/\eps))$.  

There are other constructions for $\eps$-samples, which either focus on small space (to work in a stream)~\cite{STZ04,BCEG07} or have better performance in practice without size guarantees below that of random sampling~\cite{ABM06}.  As with the optimal algorithms, these require the enumeration of all combinatorial halfspaces associated with a set of size roughly the size of the final $\eps$-sample, requiring at least $\Omega((1/\eps^{2d/(d+1)})^d)$ time.  Indeed Suri \etal~\cite{STZ04} concludes with:
``\emph{The high computational complexity of the currently known algorithms for these subroutines may be prohibitive for data stream applications. It is a long standing open problem to find efficient exact or approximation algorithms for either of them.}''

A \emph{partition} of $(X, \c{H}_d)$ is a set of pairs $\{(\Delta_1, X_1), (\Delta_2, X_2), \ldots \}$ where each $\Delta_i$ is a small complexity region and contains $X_i \subset X$, and $X$ is the disjoint union of the $X_i$s.  It is a \emph{$(t,z)$-partition} when there are $O(t)$ pairs, $|X_i| \leq 2n/t$; and each $h \in \c{H}_d$ crosses $O(t^z)$ cells.  
The smallest possible guarantee for $z$ is $z=(1-1/d)$, and an algorithm for such a construction was provided by Matou\v{s}ek~\cite{Mat92}, that takes $O(n \log t)$ time after $O(n^{1+\eta})$ preprocessing time for any $\eta > 0$.  Chan provided a refined algorithm which takes $O(n \log t)$ time, and has a few nicer structural properties.    
There are other algorithms which generate $(t,z)$-partitions for large values of $z$.  For instance in $\R^2$ Edelsbrunner and Welzl~\cite{EW86} describe an algorithm with $z=0.695$ and a structure similar to a kd-tree leads to a size of $z= \log_4(3) \leq 0.7925$~\cite{Wil82}.

\subparagraph*{Our results.}
In this paper, we use partition construction algorithms to efficiently create $\eps$-samples for $(X,\c{H}_d)$.  Our algorithm takes $O(n + \frac{1}{\eps^2} \log \frac{1}{\eps})$ time and produces an $\eps$-sample of size $O((1/\eps)^{2d/(d+1)} \log^{d/(d+1)}(1/\eps))$, nearly matching the $\Omega(1/\eps^{2d / (d + 1)})$ lower bound.  

We also implement several variants of these algorithms in $\R^2$.  We know of no other implementation of $\eps$-sample construction for $(X,\c{H}_2)$ which is guaranteed to get subquadratic size in $1/\eps$.  We know of no implementations of optimal partitions, although Har-Peled~\cite{Har00} has implemented a related concept called a cutting, which (as we will explain) is a key ingredient for creating partitions.  We choose to build our own implementation of cuttings, and explain why we did not use Har-Peled's in Section \ref{sec:implement}.  

We are able to demonstrate that our algorithm indeed scales linearly in $n$, scales linearly in the output size, and produces $\eps$-samples with less measured error than random samples.  

Our initial goal in fast $\eps$-sample construction comes from finding approximately maximal ranges in range spaces, as part of a large-scale spatial anomaly detection framework~\cite{SSSS,Scanning}.  At a high level, these algorithms follow two phases:  (1) create an $\eps$-sample $S$,  (2) use $S$ to find an approximately maximal range.  The second step takes $O(|S|/\eps)$ or $O(|S|/\eps^2)$, so it is only worth using a smaller $\eps$-sample of size roughly $1/\eps^{4/3}$ if it takes less than $1/\eps^{2+1/3}$ or $1/\eps^{3+1/3}$ time to create.  We show this is the case in theory, and in practice.  
Similar overall runtime gains exist when using $S$ for classification, or approximate range counting, or other tasks where the use of $S$ is more expensive than the new construction time.

\section{Overview and Proof for Fast $\eps$-Samples}
\label{sec:sampling-half-space}

The key to our construction of an $\eps$-sample $S$ for a range space $(X, \c{H}_d)$ 
is to first create a partition over $(X,\c{H}_d)$.  
Given such a partition algorithm, our algorithm constructs an $\eps$-sample as follows.  Randomly sample $Y \subset X$, construct the partition $\Delta = \{(\Delta_1, Y_1), \ldots, \}$ on $Y$, and return a single point at random from each $Y_i$ weighted by $|Y_i|$.

\begin{theorem}
For range space $(X,\c{H}_d)$ with $|X|=n$ and constant $d$, 
with constant probability an $\eps$-sample $S$ of size 
$O(\frac{1}{\eps^{2d/(d + 1)}} \log^{d/(d + 1)} \frac{1}{\eps})$ 
can be constructed in 
$O(n + \frac{1}{\eps^2} \log \frac{1}{\eps})$ 
time.
\label{thm:halfplanesample}
\end{theorem}
\begin{proof}
Take a random uniform sample $Y \subset X$ of size $s = O(\frac{1}{\eps_1^2})$ then $Y$ is an $\eps_1$-sample of $(X,\c{H}_d)$ with constant probability.  Next we build a $(t,1-1/d)$-partition on $Y$ in $O(s \log t)$ time~\cite{Chan10}; this results in a set of $O(t)$ partitions of $Y$ each containing at most $2s/t$ points such that any halfspace in $\c{H}_d$ will only cross $O(t^{1 - 1/d})$ of them. 
From each partition $(\Delta, Y_i)$ we will choose a single point $y_i$ at random to put in our result $S$, and weight it proportional to the number of points in the partition.

In our construction any partition contained completely inside a halfspace or outside does not contribute to the error of the sample. Only regions crossing the boundary of the halfspace $h$ contribute to the error. 
The error in each boundary region is an independent bounded random variable $V_i$ with value in the range $[0, 2\frac{s}{t}]$. 
There are at most $k = c \cdot t^{1 - 1/d}$ boundary regions for some constant $c$, so we can apply Hoeffding's inequality, with failure probability $\delta$ 
\[
 \Pr[|V - \E[V]| \ge  s \eps_2] 
	\le 
	2 \exp \left ( -\frac{2 s^2 \eps_2^2}{c t^{1 - 1/d} \cdot 4 \frac{s^2}{t^2}} \right )
	= 
	2 \exp \left (- \frac{\eps_2^2 t^{1 + 1/d} }{2c} \right ) \le \delta. 
\]
Rearranging the last inequality, gives that with $t \geq (\frac{2c}{\eps_2^2} \ln \frac{2}{\delta})^{d/(d+1)}$, 
for any one halfspace $h$, $|V - \E[V]|$ is more than $s \eps_2$ with probability at most $\delta$.

There are $O(s^d) = O(1/\eps_1^{2d})$ halfspaces in $(Y, \c{H}_2)$, so setting $\delta = c_2 \eps_1^{2d}$ for some constant $c_2$, and the additivity property of $\eps$-approximations~\cite{Cha01}, gives an $(\eps_1 + \eps_2)$-approximation of size 
$t \ge \left (\frac{4d c}{\eps^2_2}\ln \frac{2}{c_2\eps_1} \right)^{d / (d + 1)}$ 
with constant probability. By setting $\eps_1 = \eps_2 = \frac{\eps}{2}$ the total error is $\eps_1 + \eps_2 = \eps$ and the size of the $\eps$-sample is 
$O\left (\frac{1}{\eps^{2d/(d + 1)}}\log^{d/(d + 1)}\left  (\frac{1}{\eps} \right)\right)$ 
for constant $d$. 
Creating $Y$ takes $O(n + \frac{1}{\eps^2})$ time, the partition tree construction takes $O(\frac{1}{\eps^2} \log \frac{1}{\eps})$ time since $t =O(\mathsf{poly}(\frac{1}{\eps}))$, and the re-weighting and sampling step takes $O(\frac{1}{\eps^2})$ time. 
In total therefore the entire algorithm takes 
$O(n + \frac{1}{\eps^2} \log \frac{1}{\eps})$ time.
\end{proof}

The same proof technique will work with other $(t,z)$-partitions in place of Chan's~\cite{Chan10}.  
In general, for $z < 1$, a scheme that generates a $(t,z)$-partition of $t$ cells where any halfspace crosses at most $O(t^z)$ of the cells results in an $\eps$-sample of size $O(\frac{1}{\eps^{2 / (2 - z)}} \log^{1/(2 - z)} \frac{1}{\eps})$. 
For instance in $\R^2$, Edelsbrunner and Welzl's $z=0.695$ result~\cite{EW86} in an $\eps$-sample of size $O(\frac{1}{\eps^{1.532}} \log^{0.766} (\frac{1}{\eps}))$.  
Alternatively, Willards $z = 0.7925$ result in $\R^2$~\cite{Wil82} results in an $\eps$-sample size of 
$O(\frac{1}{\eps^{1.657}} \log^{0.829} (\frac{1}{\eps}))$.

\section{Overview of Algorithms for Constructing the Partition}
\label{sec:partition}

Random sampling, and sampling a point from each cell of a partition is straight-forward; the challenge in our implementation of Theorem \ref{thm:halfplanesample} is the creation of a partition.  In this section we describe the key components of the two prominent optimal size $(z=1-1/d$) algorithms:
Matou\v{s}ek's efficient partitioning~\cite{Mat92} (\pmat) and 
Chan's Optimal partitioning~\cite{Chan10} (\pchan). 

These algorithms rely on a related object called a cutting, defined over $\R^d$ and a set of $m$ hyperplanes $H$.  For a parameter $r < m$, a $(1/r)$-\emph{cutting} is a decomposition of $\R^d$ into $O(r^d)$ cells $\Lambda = \{\Lambda_1, \Lambda_2, \ldots\}$, so no cell is crossed by more than $O(m / r)$ hyperplanes in $H$.  Such cuttings exist and can be computed in $O(m r^{d-1})$ time~\cite{CF90,Mat91}.  

Cuttings are almost enough to compute partitions.  A set of $n$ points in $\R^d$ induces $m = O(n^d)$ combinatorially distinct halfspaces $H$.  Letting $r = t^{1/d}$, the total number of crossings will be $O(r^d \cdot m/r) = O(m r^{d-1})$, so the \emph{average} per region will be $O(r^{d-1}) = O(t^{1-1/d})$.  Also, ignoring dependences, the average cell contains $O(n/r^d) = O(n/t)$ points, as desired.  
The main challenge is ensuring that these average properties of the cutting map to the specific properties required for the partition.  In short, we can create an appropriate cutting, detect where it does not satisfy the partition properties, and then amend it so it does.  

We specifically focus our implementations in the $d=2$ setting, which for instance is enough for our original application of spatial anomaly detection we mentioned previously~\cite{Scanning}, even in higher dimensions.  Our implementations are similar to the existing implementation of cuttings by Har-Peled~\cite{Har00}, but adds several features which will aid in computing the partition. Our cutting implementation builds a cutting by iteratively adding lines in a random order while keeping track of the number of lines crossing each cell in an arrangement.  From a practical point of view, it is important to force the cells of the partition to be constant size.  We have focused on two methods for this, a vertical trapezoidal decomposition (\textsf{Trapezoid}), or a hierarchy of constant size polygons (\textsf{PolyTree}).

Constructing a $(1/r)$-cutting over the entire set of $O(n^d)$ halfspaces would lead to a runtime of $O(n^d r^{d- 1})$ which would be prohibitively slow. Instead of using the full set of halfspaces a smaller set (a test set) can be constructed, such that the number of partitions crossed by any halfspace in this test set will not be too different from the full set $\c{H}_d$.  

In particular, an \emph{$(1/r)$-test set} is a set of halfspaces $H$ which applies to any partition $\Delta = \{(\Delta_1,X_1), (\Delta_2,X_2), \ldots \}$ and point set $X$ of size $n$ so $|X_i| \geq n / r$ for all $(\Delta_i,X_i) \in \Delta$. 
 It ensures that if $\kappa = \max_{h \in H} |h \cap \Delta|$, then $\max_{h \in \c{H}_d} |h \cap \Delta| \leq O(\kappa + r^{1-1/d})$.  Here $h \cap \Delta$ is the set of $(\Delta_i,X_i) \in \Delta$ for which $\Delta_i$ intersects $h$, but do not completely contain $h$.  	
Test sets can be built a number of ways, including randomly sampling lines, randomly sampling points and using the lines they induce, and using the dual arrangement.

\section{Implementation Particulars of Partitions}
\label{sec:implement}

Our implementation of Partition trees is in python.  It relies on an efficient way to construct and maintain an arrangement of lines and associated points.  At each step of the construction we 
will maintain a tree with leaves that correspond to cells $\Delta_1, \Delta_2, \ldots$ of an arrangement.  Each cell will maintain a list of contained points $X_i \in \Delta_i$ and crossing lines.  

As part of the construction so the result is a $(t,1-1/d)$-partition $\Delta$, with desired $t$ parameter, cells can be refined by applying various operations to them.  For instance a cutting can be constructed locally inside of a cell $\Delta_i$, or a cell can be partitioned into a set of sub-cells.  
  
\subparagraph*{Geometric Primitives.} 
All of our algorithms rely on operations over line segments. The most important operation is being able to test, within a region $\Delta$, if a line lies completely above a line segment or if it crosses a line segment. 
This fairly simple operation is slightly complicated by numerical issues that can occur. 
For instance when constructing a test set using the \ptestset or \dtestset method (see below)
many lines will potentially meet at the same point. Line segments that meet in this 
point could be mistaken as crossing. To handle numerical issues we use python's
implementation of \texttt{math.isclose} to handle point comparisons. This method allows us to assign two floating point numbers as equal if their relative values are sufficiently close~\cite{PEP}.  
Moreover, all methods that compare line segments have closed and open versions where closed versions allow end-point overlap and open versions do not. The method \texttt{segment.above\_closed(line)} returns true if the \texttt{line} intersects with the segment at one the segment's end points, but is otherwise above the segment, while \texttt{segment.above\_open(line)} returns false in this case.  This allows us in our experiments to effectively handle degeneracies while avoiding slower exact precision libraries.

Internally our segment objects are represented by the slope, $a$, the 
$y$-intercept, $b$, and the $[x_l, x_r]$ interval on the $x$-axis the segment is defined over. 
This representation makes many operations easy, but also results in several challenges, most notably: vertical lines are undefined, unbounded segments (e.g., $(-\infty, x_r]$) require extra logic to handle crossing queries, lines which are nearly vertical can become numerically unstable, and the dual of unbounded polygons require significant extra logic to handle correctly.  However, we have implemented stable functions for intersect and above relations for pairs of segments in a cell.  

Using line segments and points as the primitives we also define more complicated 
structures notably: polygons, dual wedges, vertical line segments, and trapezoids.

\subparagraph*{PolyTrees.} 
There are a number of ways to maintain the structure of an arrangement.  A common method is to store each cell with a corresponding list of pointers to adjacent cells. Inserting a line involves finding the leftmost crossing cell, 
identifying the next adjacent cell the line crosses, splitting the crossed cell into an 
upper and lower cell, and then repeating this operation for each crossed cell.

This has a number of downsides: there are special cases if a line crosses a vertex of 
a cell, inserting points into the arrangement requires the maintenance of a secondary structure, and cells require a significant amount of adjacency information that must be maintained. 
Instead of maintaining this structure we use the idea of forcing each cell to be simple, and follow certain restrictions, as introduced by Seidel~\cite{seidel1991} and refined by Har-Peled~\cite{Har00}.  
 
In particular, we either maintain a decomposition into constant complexity polygons (polygons with a constant number of boundary segments) or a trapezoidal decomposition.  In both cases we maintain a tree where each node in the tree consists of a line segment that separates a cell into two cells. 
With trapezoids an inserted line could in some cases divide a trapezoid vertically into 2 separate trapezoids and then horizontally into 4 separate trapezoids. 
In the case of polygons the inserted line would split the polygon into two separate polygons which could possibly be further split if the number of sides in either of the resulting polygons is greater than a chosen constant.  We also enforce that no vertical segments are used to avoid limitations of our line segment representation.  

Given a line $h$ and a decomposition $\Lambda = \{\Lambda_1, \Lambda_2, \ldots\}$, the \emph{zone} of $h$ is the set of regions $\Lambda_i$ that intersect $h$; we represent this as $\zone_h = \Lambda \cap h$.  
To find the zone of a line in this structure at each node we treat the line as an infinite length segment and then traverse the line down the tree. At each node we will have three cases where the portion of the line contained in the node lies completely either above or below, or crosses the current node's line segment. In the completely above or below case we merely traverse to the above or below child of the node.  In the crossing case we split the portion of the line contained in the node into two segments, above and below, and recursively query the above and below nodes. 
Point information is easy to maintain with this method since a point always lies on one side of the line segment, so the tree structure can be used to insert or remove points in logarithmic time to the number of cells.

More complicated structures can also be queried on these trees, most notably wedges and  polygons. Wedge queries are particularly useful in \pchan since a wedge is the dual of a line segment, so the number of points contained in a wedge corresponds in the dual to the number of lines crossing a line segment.  
\begin{algorithm} 
  	\caption{$\cutting(H,r)$} 
  	\label{alg:cutting} 
  	\begin{algorithmic}[1] 
  		\STATE $\Lambda = \R^2$
  		\FOR{$h \in H$ (ordered by a random weighted permutation)}
  		\STATE Find $\mathsf{Viol}(h,\Lambda) = \{\Lambda_i \in \zone_h(\Lambda) \mid |H \cap \Lambda_i| > |H|/r\}$.  
  		\STATE For all $\Lambda_i \in \mathsf{Viol}(h,\Lambda)$, replace $\Lambda_i$ in $\Lambda$ by $\Split(\Lambda_i,h)$
  		\ENDFOR
  		\RETURN $\Lambda$
  	\end{algorithmic}
\end{algorithm}

\vspace{-.1in}
\subparagraph*{Cuttings.} 
Our cutting algorithm $\cutting(H,r)$ (Algorithm \ref{alg:cutting}) follows closely Algorithm \ref{alg:cutting}, from Har-Peled~\cite{Har00}. 
We implement the cutting with respect to weighted lines as this speeds up and somewhat simplifies the later partitioning algorithms. We require a weighted permutation of lines using \cite{ES06}; this ensures that the probability we see a line after some point in the permutation is equivalent to the probability we would have seen at least one instance after seeing that many distinct lines in a variant where weights are multiplicities (as advocated by Chan~\cite{Chan10}), and each copy is treated independently in a uniform random permutation.  
For notational convenience, for a subset $H' \subseteq H$, let $|H'| = \sum_{h \in H'} w(h)$, where $w(h)$ is the weight implicitly stored with each halfspace $h \in H$.  

\begin{figure}[b]
	\vspace{-.18in}
 	\includegraphics[width=0.32\linewidth]{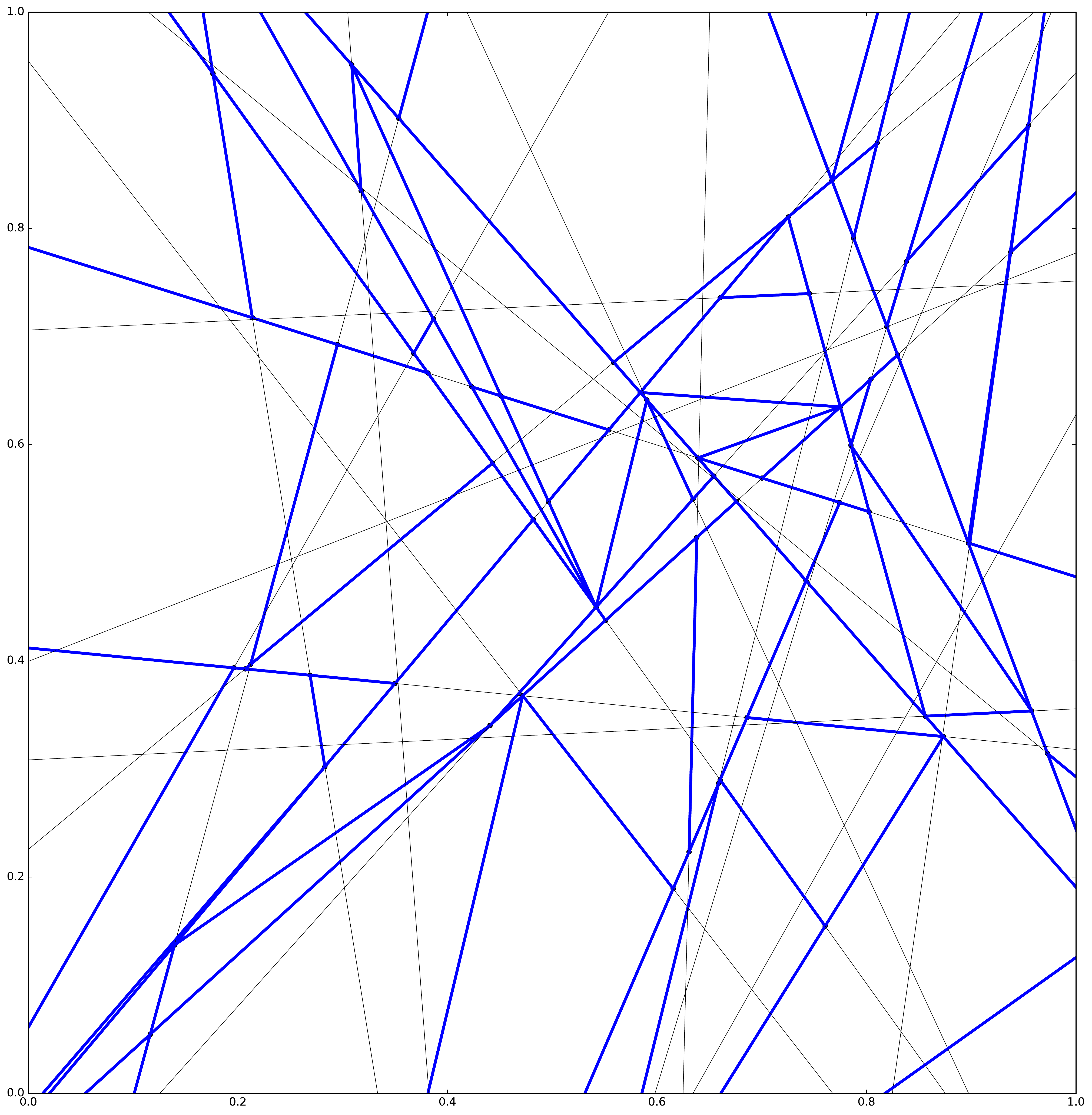}
 	\includegraphics[width=0.32\linewidth]{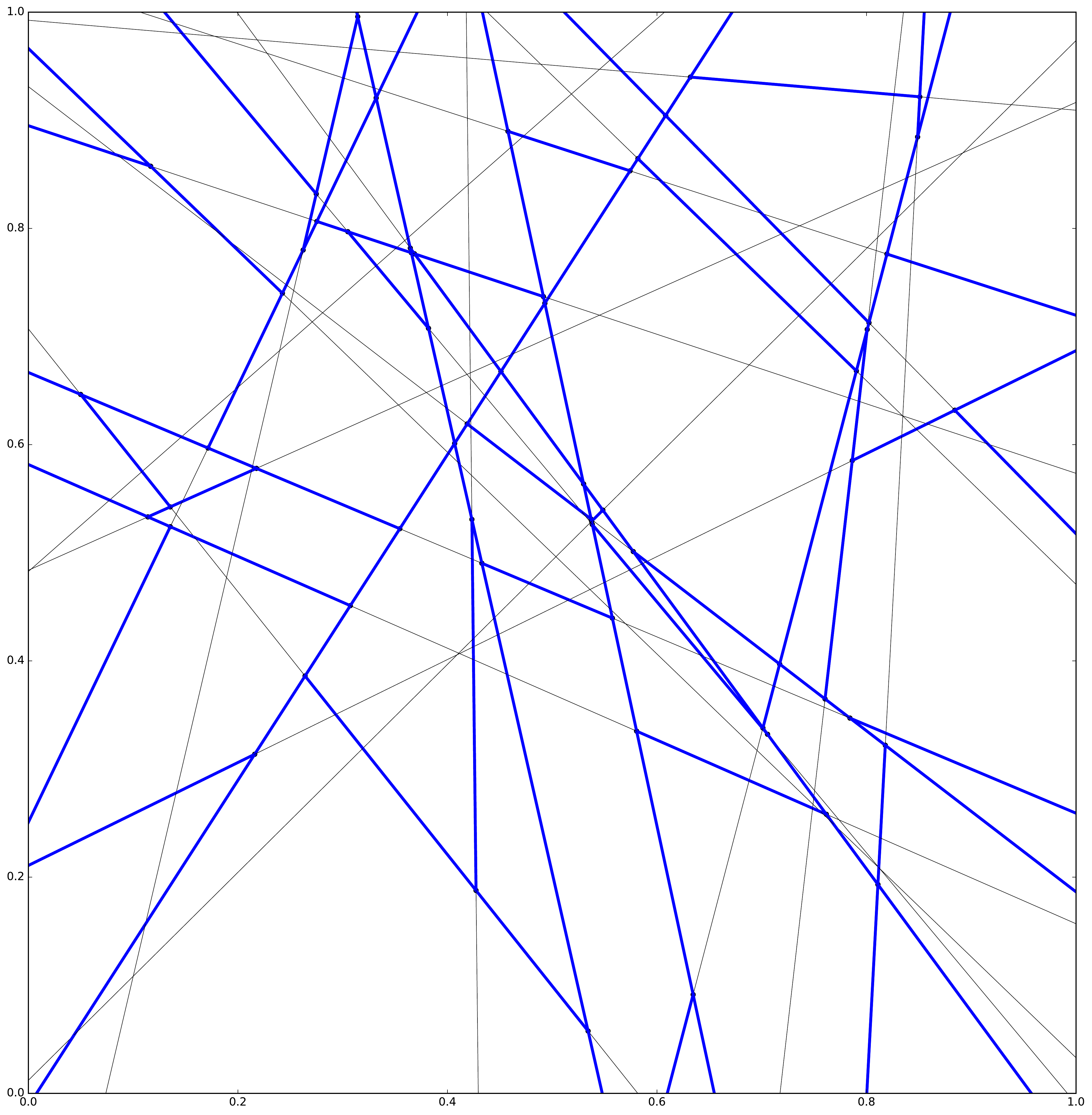}
 	\includegraphics[width=0.32\linewidth]{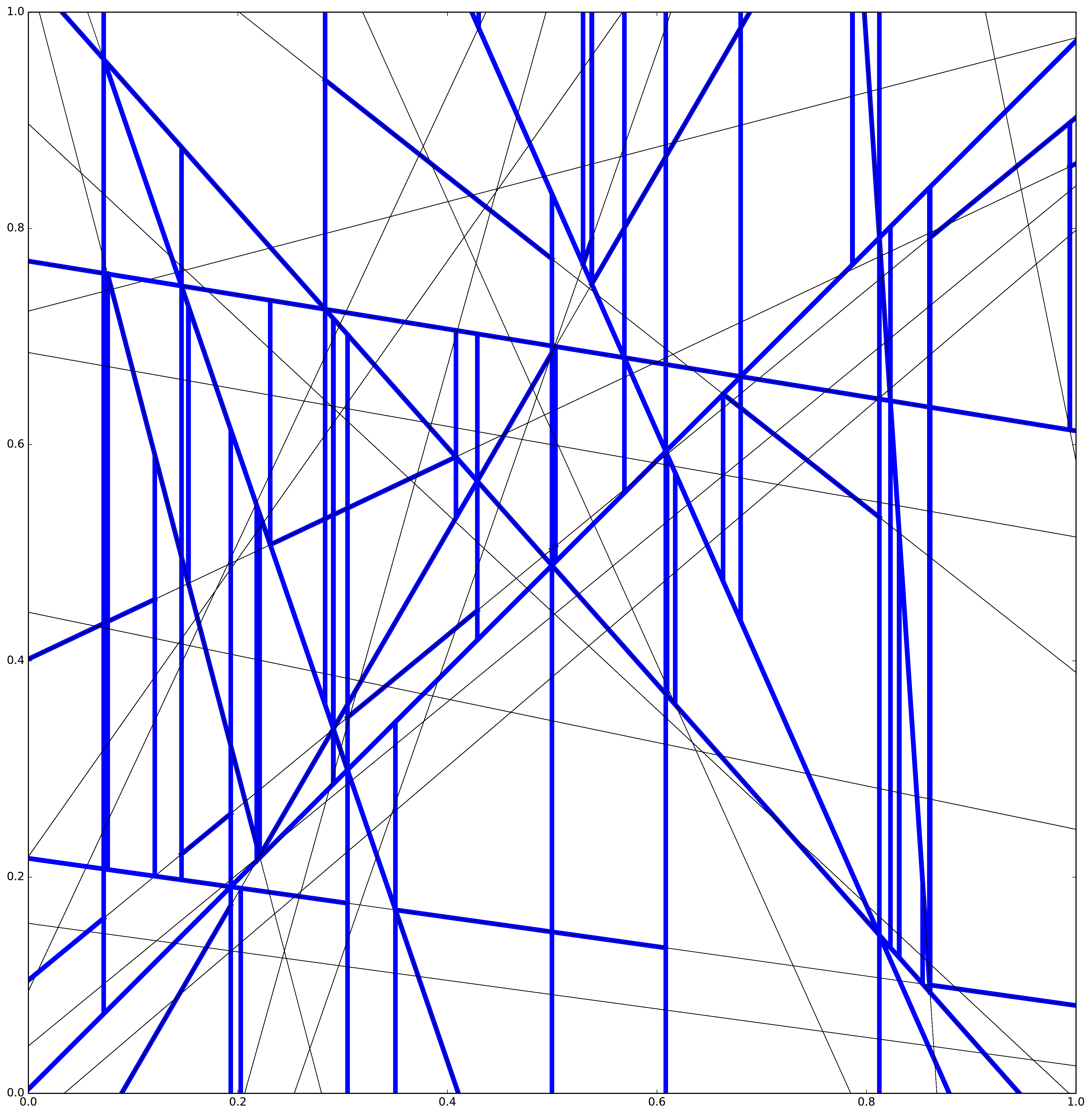}
 	\vspace{-.12in}
 	\caption{Cuttings with 25 lines and at most 5 lines crossing each cell. From left to right \polyT with at most 4 sides, \polyT with at most 8 sides, and \trap.}  
 	\label{fig:cutting_diag}
\end{figure}
In practice, we implement $\cutting(H,r)$ slightly differently then described in Algorithm \ref{alg:cutting}.  Instead of choosing the $h \in H$ to process in the random weighted permutation, our approach is centered around the violated cells.  We choose a cell $\Lambda_i \in \Lambda$ which is too heavy (i.e., $|H \cap \Lambda_i| > |H|/r$), and then choose some halfspace $h \in \Lambda_i \cap H$, and only replace $\Lambda_i$ (not the entire $\zone_h$) with the result of the $\Split(\Lambda_i, h)$, which divides $\Lambda_i$ into two parts separated by $h$.  This change has two advantages.  First we do not need to find the $\zone_h(\Lambda)$ which involves traversing the \polyT, so our approach is slightly faster.  Second, we can choose the $h \in \Lambda_i \cap H$ to use in the split wisely; e.g., as the one that maximizes the smaller of the two resulting cells.  We find the second heuristic produces slightly smaller cuttings in practice, but is significantly slower, and is not used in our experiments.    
	
How $\Split(\Lambda_i, h)$ is implemented is the difference between the \trap-based cutting and the \textsf{Polygon}-based cutting we refer to as \polyT.  

For the \trap-based method each cell $\Delta$ is a trapezoid to begin with. The $\Split$ operation first inserts up to two new vertical cuts for each intersection of the line with the top or bottom of the cell and then horizontally cuts the resulting cells using the inserted line. 
For \polyT, the first important detail is that we store $H \cap \Lambda_i$ as the line segments restricted to where they intersect $\Lambda_i$.  

On a split, we need to maintain which halfspaces $h \in \Lambda_i \cap H$ are in each child; if $h' \in \Lambda_i \cap H$ intersects both children, then we split $h'$ into two line segments at the point where it intersects $h$, and store the corresponding segment in each child.  If $h' \in \Lambda_i \cap H$ is in only one child, because we store them as segments it is easy to check which child it goes into.  

Both of these algorithms are then fairly straightforward to implement once given structures for efficiently maintaining arrangements of line segments.
  
We find that \polyT is faster and produces a smaller cutting than \trap-based ones; see Figure \ref{fig:r4cuttingsizesmall} and Figure \ref{fig:cuttingsize} which show the runtime and cutting size as a function of input size and choice of $r$. 
For this reason we will primarily focus on \polyT hereafter.  
  
  \begin{figure}[t]
  	\includegraphics[width=0.49\linewidth]{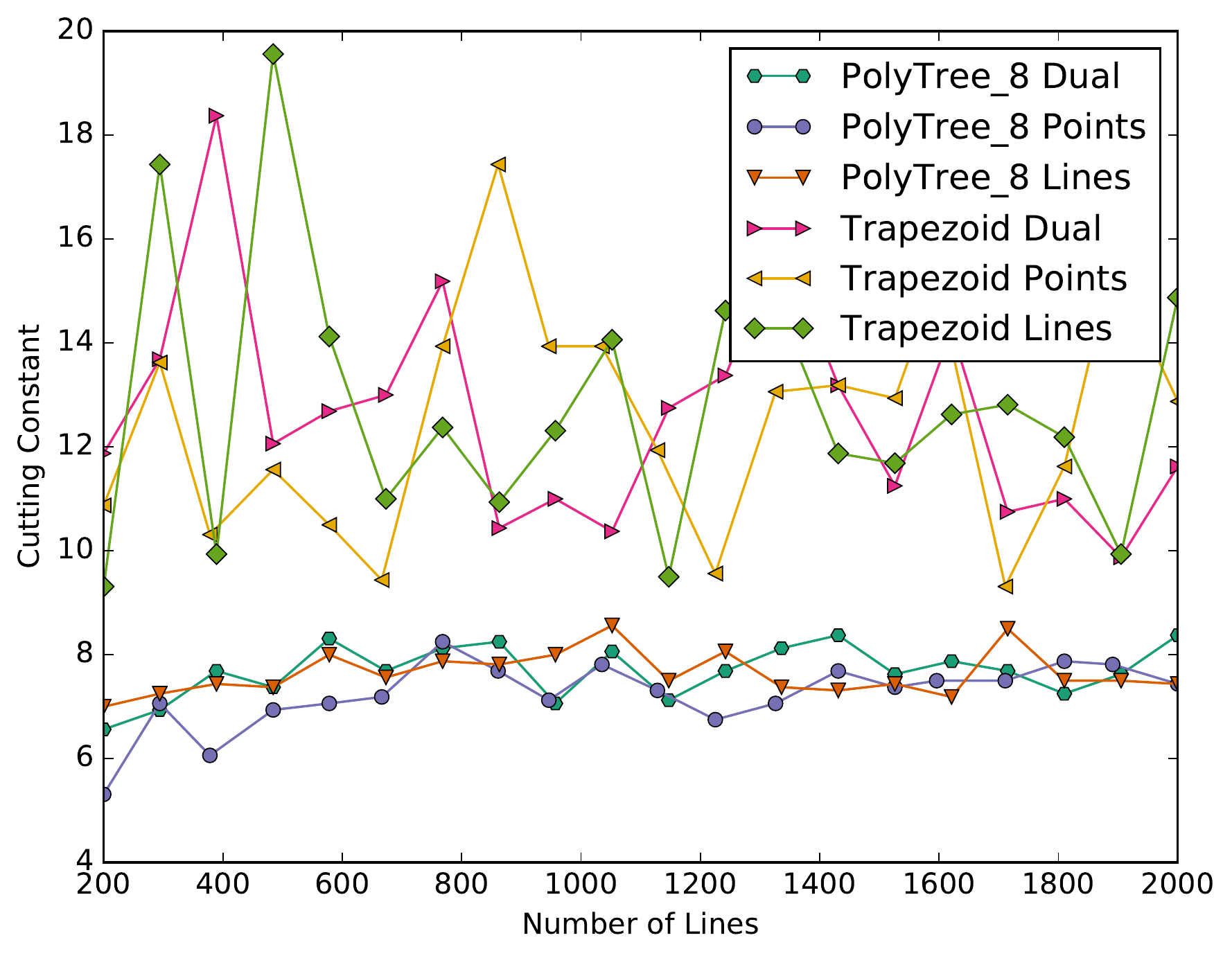}
	\includegraphics[width=0.49\linewidth]{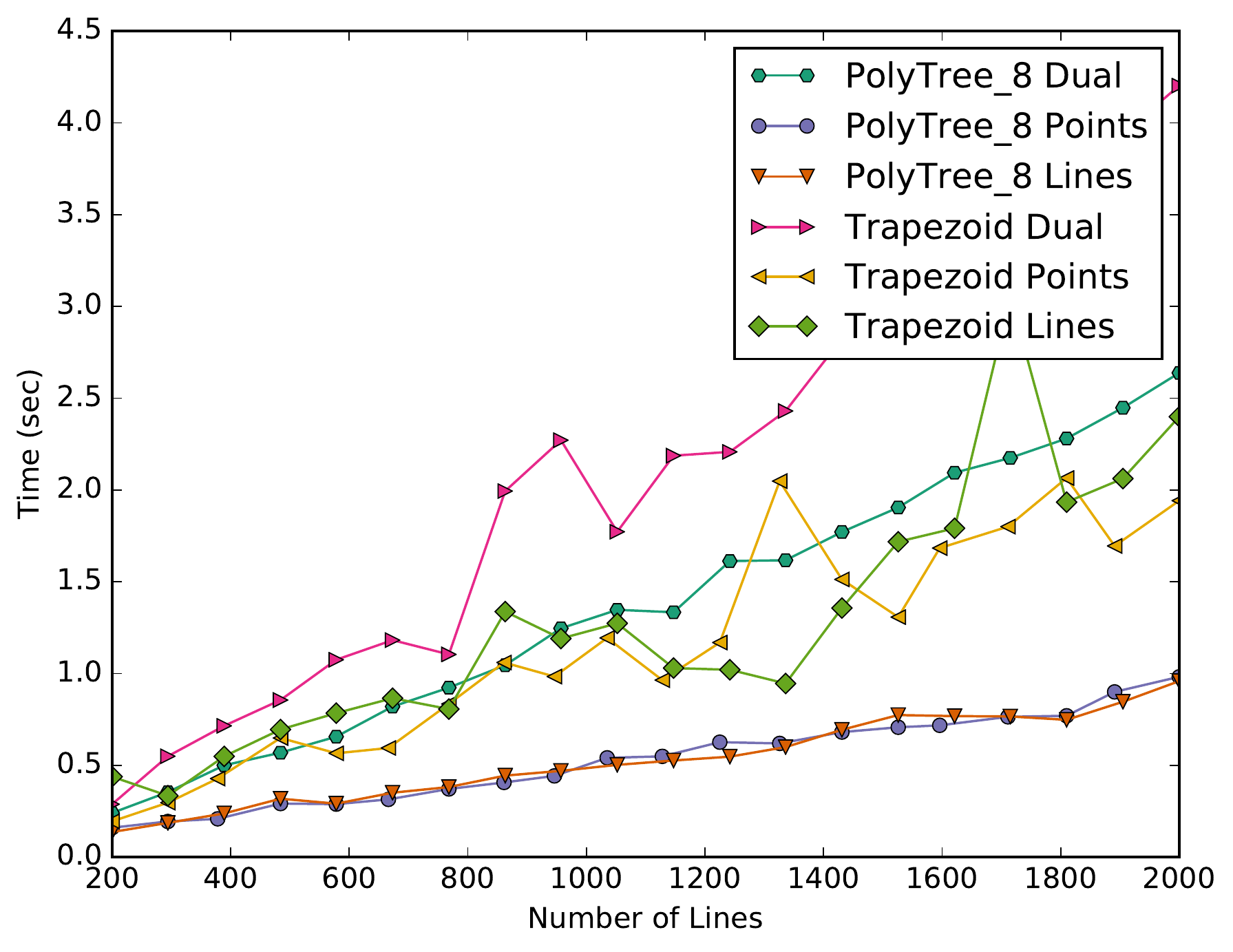}
	
	\vspace{-.18in}
  	\caption{Size of cutting (divided by $r^2$) and time (in seconds) vs. the number of input lines.}  
  	\label{fig:r4cuttingsizesmall}
  \end{figure}
  \begin{figure}[t]
     \vspace{-.2in}
  	\includegraphics[width=0.49\linewidth]{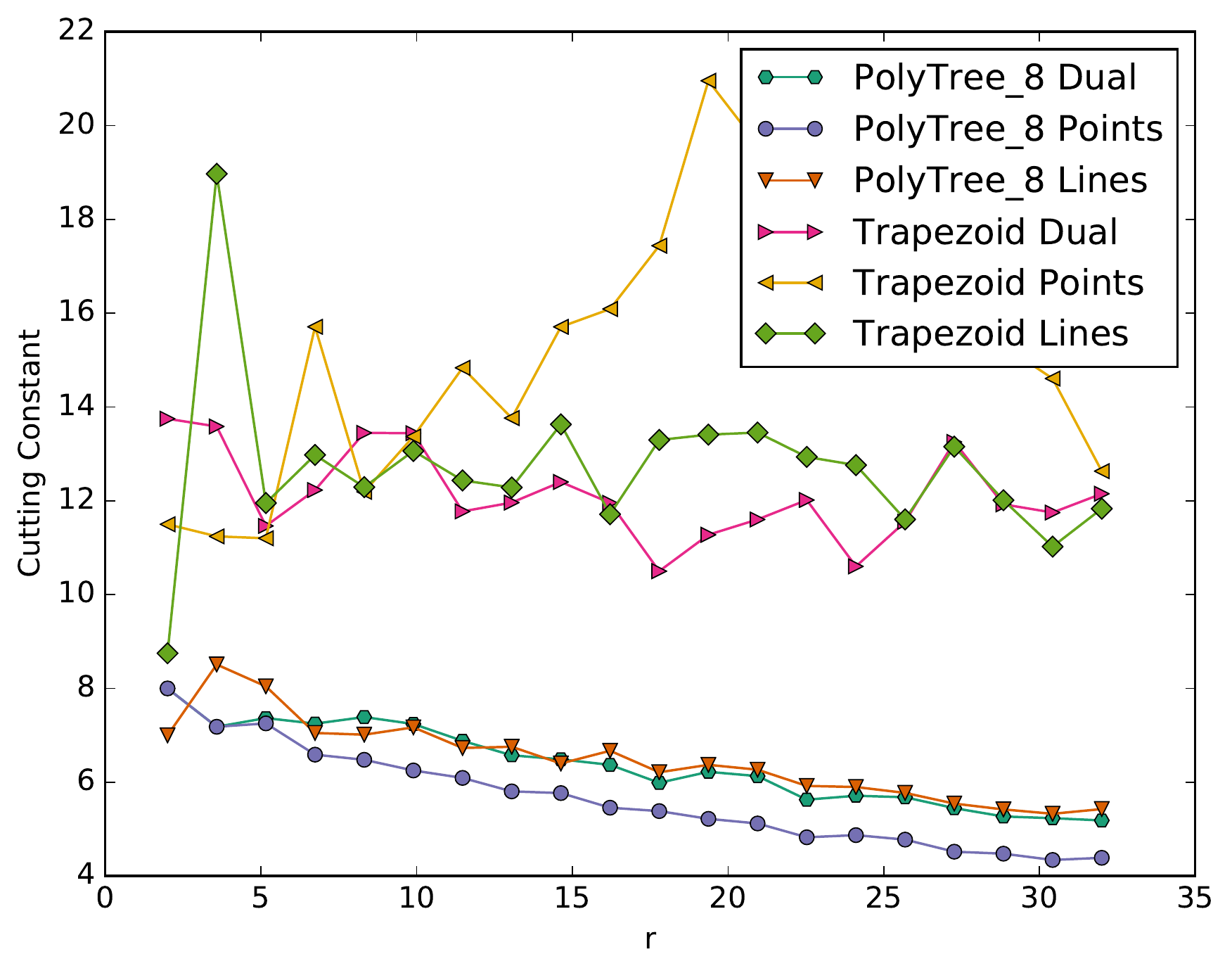}
  	\includegraphics[width=0.49\linewidth]{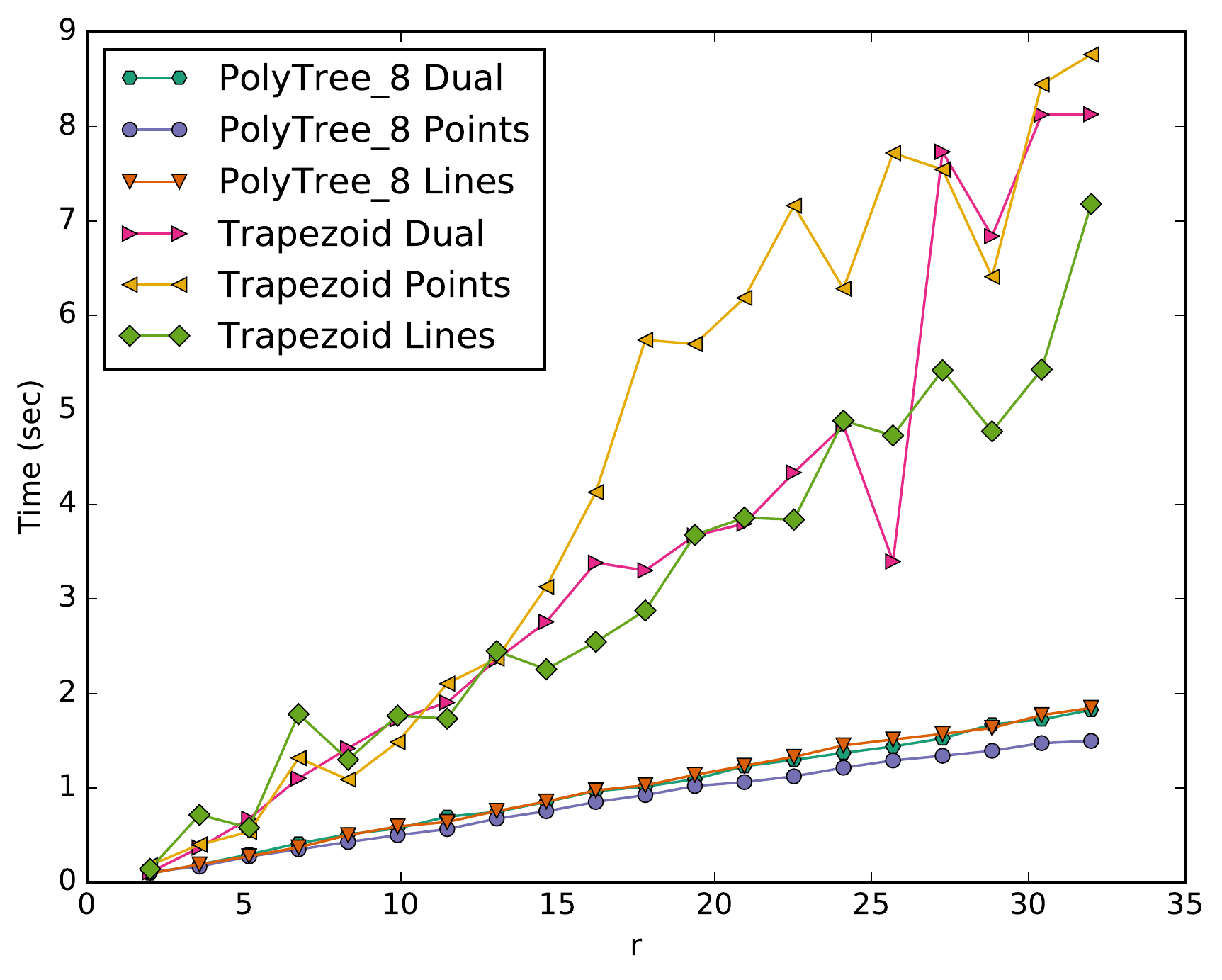}

	\vspace{-.18in}
  	\caption{Size of cutting (divided by $r^2$) and time (in seconds) as $r$ increases for $1000$ input lines.}
  	\label{fig:cuttingsize}
  \end{figure}

\begin{algorithm} 
	\caption{\dtestset$(X, r)$} 
	\label{alg:dualtest} 
	\begin{algorithmic}[1] 
		\STATE $S = \textrm{sample}(X, O(\sqrt{r}\log r))$;  $S^*$ is dual of $S$.   
		\STATE $\Lambda \leftarrow \cutting(S^*, O(r^{1/2}))$ 
		\RETURN the dual of $V^*$, where $V$ is the set of vertices of the cells of $\Lambda$.  
	\end{algorithmic}
\end{algorithm}

\vspace{-.1in}
\subparagraph*{Test Set Generation.}
There are a number of ways to generate test sets (\dtestset, \ptestset, \ltestset), but these do not appear to have a significant effect on the runtime of the final algorithms; again see Figure \ref{fig:r4cuttingsizesmall} and Figure \ref{fig:cuttingsize} for comparisons of the \textsf{PolyTree} and \textsf{Trapezoid} methods.   
The simplest method, \ltestset, simply samples $O(r \log^d n)$ halfplanes, from those defined by passing through $d$ points in $X$.  
The next simplest, \ptestset, samples $O(r^{1/d} \log n)$ points $S$, and then the test set is all halfplanes passing through $d$-tuples chosen from $S$; it again defines $O(r \log^d n)$ halfplanes.  
Finally the most complicated approach is \dtestset (see Algorithm \ref{alg:dualtest}); it produces the smallest size test set, size $O(r)$~\cite{Mat92}, and thus is the one we advocate.  It samples $O(r^{1/d} \log r)$ points $S \subset X$; it considers the dual set of halfplanes $S^*$ of primal points $S$; it creates a $(1/r^{1/d})$-cutting of $S^*$ (in the dual); and then it returns the primal halfspaces defined by the vertices of the cutting in the dual.  
Each halfspace $h$ in the test set $H$ is implicitly endowed with a weight $w(h)$, which by default is $w(h)=1$ for all $h \in H$.  

\begin{algorithm}[h] 
	\caption{$\pmat(X, r, n, j)$} 
	\label{alg:mpart} 
	\begin{algorithmic}[1] 
		\STATE \textbf{if} ($|X| < n / r$) \textbf{then} \textbf{return} $\{(X,\Delta_0)\}$ where $\Delta_0$ contains $X$.  
		\STATE $H \leftarrow \gentestset(X,b/2^j)$ 
		\STATE $\Delta = \emptyset$
		\WHILE{($|X| \geq n/2^j$)}
		  \STATE $\Lambda \leftarrow \cutting(H, \sqrt{b/2^j})$ 	
		  \STATE Find $\Lambda_i \in \Lambda$ so $|X \cap \Lambda_i| > n / b$; shrink $\Lambda_i$ so $|X \cap \Lambda_i| = \lfloor n/b \rfloor$ exactly.
		  \STATE Add $(\Lambda_i, \Lambda_i \cap X)$ to $\Delta$; remove $\Lambda_i \cap X$ from $X$
		  \STATE Double the weight $h \in H$ which cross $\Lambda_i$ 
		\ENDWHILE
        \STATE $\Delta' = \bigcup_{(\Delta_j, X_j) \in \Delta} \pmat(X_j,r,n,j)$
		\RETURN $\Delta' \cup \pmat(X,r,n, j+1)$
	\end{algorithmic}
\end{algorithm}

\vspace{-.1in}
\subparagraph*{Matou\v{s}ek Partitioning.}
We have implemented Matou\v{s}ek's efficient partition trees~\cite{Mat92}. At a high level this algorithm computes the cutting of a test set and then finds a single good cell that contain at least $n/b$ points (for a constant $b$, we use $b=16$ as default).  It adds this cell to the partition, doubles the weight of all halfspaces in the test set crossing that cell, computes a new cutting and good cell.  It repeats until the number of points remaining has been cut by half, and then it recurses on the remained of the points at half the precision (e.g., set $b := 1/2 b$).  
This is too expensive to do with $b=r$, so after this we then recursively partition each cell $(\Delta_i, X_i)$ until the result is an $(1/r, 1-1/d)$-partition as desired.  
The branching factor of the partition tree is not fixed on each level, but will be roughly $b$ on average.  
Algorithm \ref{alg:mpart} presents this approach, and is initially called as $\pmat(X, r, |X|,0)$.  

Note that Line 9 is the refinement step where each cell is further partitioned.  Since the first level is the most important for good $\eps$-samples, faster algorithms could be used at later recurve calls at this step.  
In contrast, the recursive call at Line 10 is handing objects not handled in the first pass, where each pass handles roughly half of the data.

\subparagraph*{Chan Partitioning}
Chan's optimal partition trees \cite{Chan10} are faster in theory than Matou\v{s}ek's algorithm, but are more complicated to implement.  
The algorithm works by processing each node at a certain level in the tree in a random order. For each node it creates a cutting of approximately $b/4$ size for an appropriately large branching parameter $b$ (our implementation uses $b = 22$ as a default).  It then further splits the cells of the cutting to contain $1/b$ fraction of points at that node.  It multiplicatively updates weights for halfplanes that cross each cell. This multiplicative update influences subsequent cuttings by biasing away from creating cells that are crossed by already heavily weighted lines (lines that cross many cells).  After splitting all of the cells in this level of the tree the algorithm recurses on the newly created level. 
The ultimate partition $\Delta = \{(\Delta_1, X_1), (\Delta_2, X_2), \ldots\}$ are the leaf nodes of the tree.  
Algorithm \ref{alg:cpart} presents this approach, calling $\pchan(\{(\R^2, X)\}, r, |X|)$ initially.  

\begin{algorithm} 
	\caption{$\pchan(\Delta, r, n)$} 
	\label{alg:cpart} 
	\begin{algorithmic}[1] 
		\STATE Trim to $\Delta' = \{(\Delta_i, X_i) \in \Delta \mid |X_i| > n/r\}$; \textbf{if} $\Delta' = \emptyset$ \textbf{return} $\Delta$
		\STATE $H = \gentestset(X, |\Delta|)$
			\FOR {$(\Delta_i, X_i) \in \Delta'$}
		     \STATE Sample $L \subset H$, proportional to their weight $w(h)$, at rate $q$
			\STATE $\Lambda = \cutting(L, r_i)$; with $r_i$ chosen so $|\Lambda| \leq b/4$
			 \STATE For all $\Lambda_j \in \Lambda$, further split $\Lambda_j$ (with $\Split$) until $|X_i \cap \Lambda_j| \leq |X_i|/b$
			 \STATE Replace $(\Delta_i, X_i)$ in $\Delta$ with $\{(\Lambda_1, \Lambda_1 \cap X_i), (\Lambda_2, \Lambda_2 \cap X_i), \ldots\}$

			\STATE Update all weights $w(h) = w(h) (1+1/b)^{|h \cap \bar \Lambda_i|/p}$. 
		\ENDFOR
		\RETURN $\pchan(\Delta, r, n)$ 
	\end{algorithmic}
\end{algorithm}

Implementing the algorithm as described is too slow asymptotically, so Chan presents a faster variant, which requires two additional parameters $p$ and $q$.  
Roughly $q = \sqrt{b|\Delta|}/|X|$ (see \cite{Chan10} for details) determines the probability that a line ends up in the reduced test set $L$. The parameter $p$,  about $\sqrt{b/|\Delta|} \log n$ (again, see \cite{Chan10} for details), effects the number of cells $\Lambda_i$ that are used to update the weight in each $h$  (we sample each cells with probability $p$ as opposed to dividing by this number, as written on Line 8).  
Also, Line 4, where $L$ is sampled from $H$, can be made more efficient by only minimally updating $L$ each pass through the loop, since it generally has large weight lines and that set is fairly stable.

Near the bottom of the tree, Line 8 can be expensive.  We make this efficient with a crucial observation that the test set $H$ was generated by computing a cutting over the dual space.  Thus these halfspaces are duals to the vertices of the \polyT structure.  Thus we can search over the \polyT to determine the number of crossing lines.
A cell of the partitioning is a polygon consisting of a constant number of line segments. A line crossing the polygon will cross at least one of the line segments and in the dual this will correspond to a point contained inside of a double wedge.  
For each line segment in the polygon we take its dual (a double wedge) and query the \polyT that was used to construct the test set to determine the number of vertices contained inside of it.  Since we only return the overlapped polygons and each polygon consists of at most a constant number of edges, the number of queried cells can only be a constant factor larger than the number of lines crossed by the line segment.

However, code profiling shows that the two steps involving sampling with $p$ and $q$, and updating $L$ are the most expensive parts of the algorithm.  As a result we also consider a variant $\pchanF$ which avoids these sampling steps that were supposed to speed things up.  In the context of Algorithm \ref{alg:cpart} this basically sets $p=q=1$, so $L=H$, and Line 4 is not required.

The given algorithm is only guaranteed to compute a set of 
partitions in $O(n \log^{O(1)} n)$ time; incurring extra log factors due to the height of the partition tree.  Chan removes log factors with a method he calls bootstrapping.  We do not do this since the branching factor is high (around $22$) so the depth of the tree is low, and this method is not worth the overhead.

In our implementation, we only compute the test set $H$ once at the beginning.  On each recursive call (Line 9) we can reuse it, but simply reset all of the weights to be uniform.

\subparagraph*{Why not use Har-Peled's implementation and CGAL?}
It may seem at first that we could simply use Har-Peled's implementation for $\eps$-cuttings~\cite{Har00}.  However, our initial goal was to use this as part of a code for spatial anomaly detection~\cite{SSSS,Scanning}, and there were several issues that made this less feasible.  
(1) We wanted to use non infinite precision floating point arithmetic.  Har-Peled reports that switching to exact precision representations results in a 30-factor slow down, but was necessary for degeneracy issues.  We managed these precision issues while using floating point arithmetic with careful use of open and closed operators for line above/below and intersection.  
(2) We can measure wedges, and line segments on the \polyT structure which is very useful in  \pchan.
(3)  Har-Peled's code created a cutting inside of a $1 \times 1$ box. This makes computing dual cuttings difficult as we first have to normalize the lines to lie in such a region, but computing the 
correct normalization quickly would require us to re-implement much of the \polyT algorithm. 
Ultimately we opted to build our $\eps$-cutting code from scratch rather than modify the previous code.  

Har-Peled~\cite{Har00} also reported the cutting constant (number of cells divided by $r^2$) for various of his algorithms, about $7.3$ (polygons) and $12.8$ (trapezoids).  This roughly matches the numbers we observe in  Figure \ref{fig:r4cuttingsizesmall} and Figure \ref{fig:cuttingsize}.  

Another option for computing cuttings and managing the partition trees is using the current 2d-arrangement implementation in CGAL \cite{wfzh-a2-18a}. This would have most likely made portions of this project much easier to implement and removed various hurdles. 
However, the possibilities of several factor slow-downs using exact precision would have potentially resulted in no ultimate gains in the spatial anomaly application demonstrated below.  

\section{Ham-Sandwich Tree}
\begin{figure}[t]
	\centering
	\vspace{-3mm}
	\includegraphics[width=.8\linewidth]{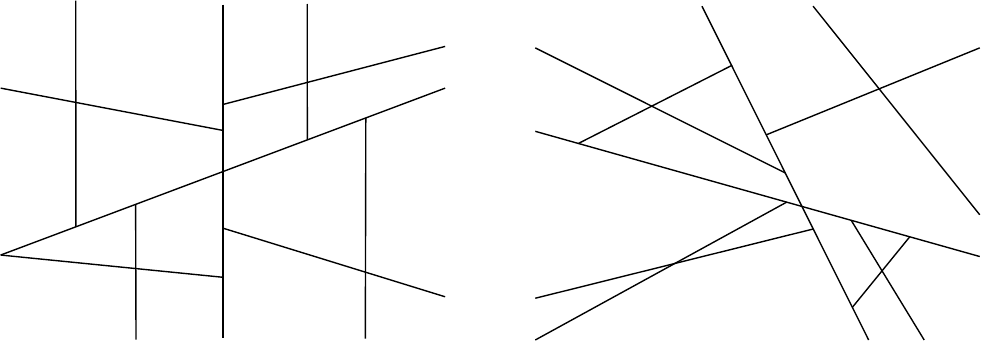}
	
	\caption{\Akd on the left (single vertical line and then ham-sandwich cut) and \Akdd on the right (ham-sandwich cut at every level).}
		\vspace{-3mm}
	\label{fig:ham}
\end{figure}
We implement two alternative partitioning methods by Willard~\cite{Wil82}, \Akd, and Edelsbrunner~\cite{EW86}, \Akdd. The first method, \Akd, provides a partitioning with $z=\log_3 4$, which gives a $O(\frac{1}{\eps^{1.657}} \log^{0.829} \frac{1}{\eps})$ sized sample, and constructs a tree with a branching factor of 4. The second method, \Akdd, provides a partitioning with $z=.695$ leading to a sample size of $O(\frac{1}{\eps^{1.533}} \log^{0.766} \frac{1}{\eps})$. Both methods are much simpler then the earlier given partitioning methods and potentially faster.

With \Akd at each level we split the point set in half with a single vertical line, and on these two resulting sets we find a single (roughly horizontal) line that divides both the left and right point set in half. Such a separator is guaranteed by the ham-sandwich theorem, and can be computed in linear time~\cite{M85}, but is complicated to implement. We instead approximate the ham-sandwich cut by computing a number of test lines and choosing the best separator from these. This is simple to implement, gives good cuts in practice, and can guarantee to be at most $\eps$-imbalanced \cite{Mat92e}. 

\Akdd differs in the way the levels are handled. We will substitute the vertical cut from above with a ham-sandwich cut. Starting with a pair of roughly equal sized point sets, each are cut in half using the ham-sandwich cut. We then recursively cut these resulting pairs again with more ham-sandwich cuts.
\begin{figure}[b]
	\vspace{-3mm}
	\includegraphics[width=0.46\linewidth]{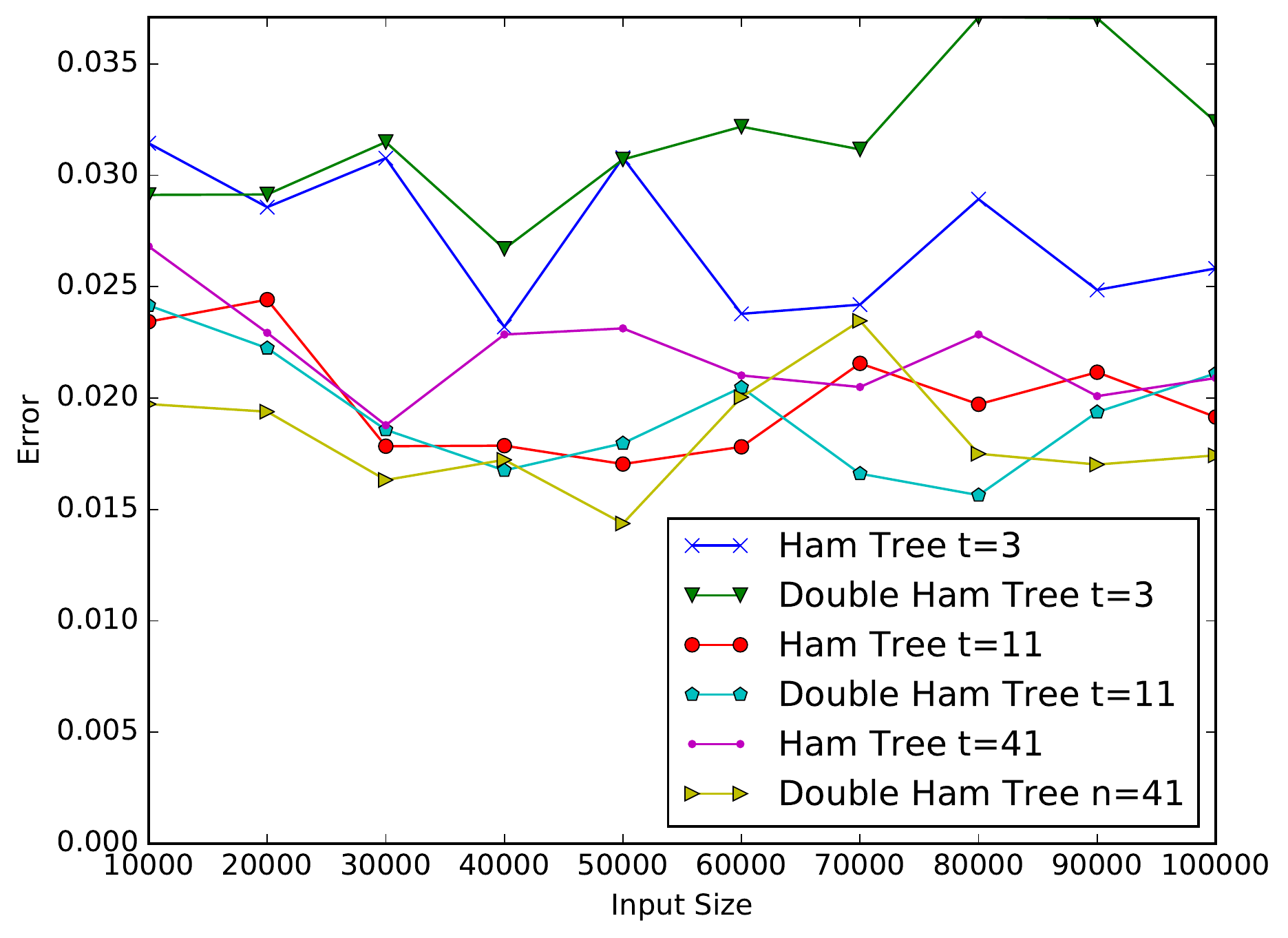}
	\hspace{3mm}
	\includegraphics[width=0.45\linewidth]{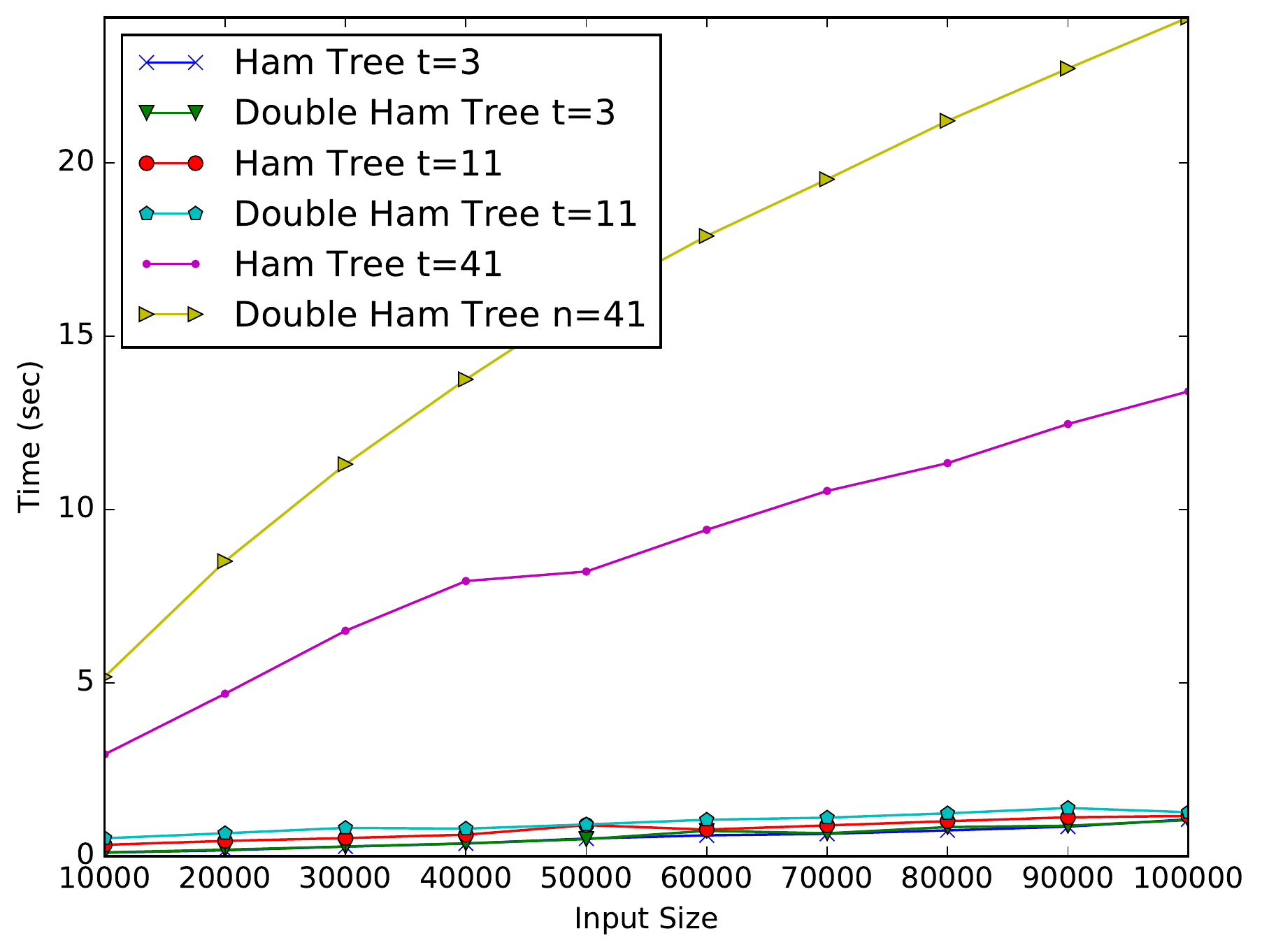}
	
	\vspace{-.18in}
	\caption{Input size vs. Time and Error for various $t$ with \Akd and \Akdd.}
	\label{fig:haminput}
\end{figure}

For both methods we compute the set of potential, approximate ham-sandwich cuts in the same way.  We randomly select $t$ points, and consider the separating line generated by those passing through each pair.  
We have experimented with the sample size $t$; as the number of test set lines increases we approximate the ham-sandwich cut with higher accuracy, but the cut also takes more time to compute.
We test various $t$ to determine a good compromise between accuracy and performance. We construct samples of the Chicago crime data~\cite{Chicago} with roughly 6.5 million data points. These experiments follow closely the setup from Section \ref{sec:exp}, see there for more details.
We find that in practice there is little difference between the accuracy and time that \Akdd and \Akd take, so in general we recommend the simpler \Akd.   We also find that smaller values of $t$ (at $t=11$) seem to provide accuracies nearly as good as large $t$ (at $t=41$), but are significantly faster.  However, quite small values of $t$ (at $t=3$) offer minimal speedup while having a minor, but observable increase in \error.  Ultimately, among these variants, we recommend \Akd with $t=11$, and this is the variant used in further experiments.

\begin{figure}[H]
	\vspace{-3mm}
	\includegraphics[width=0.46\linewidth]{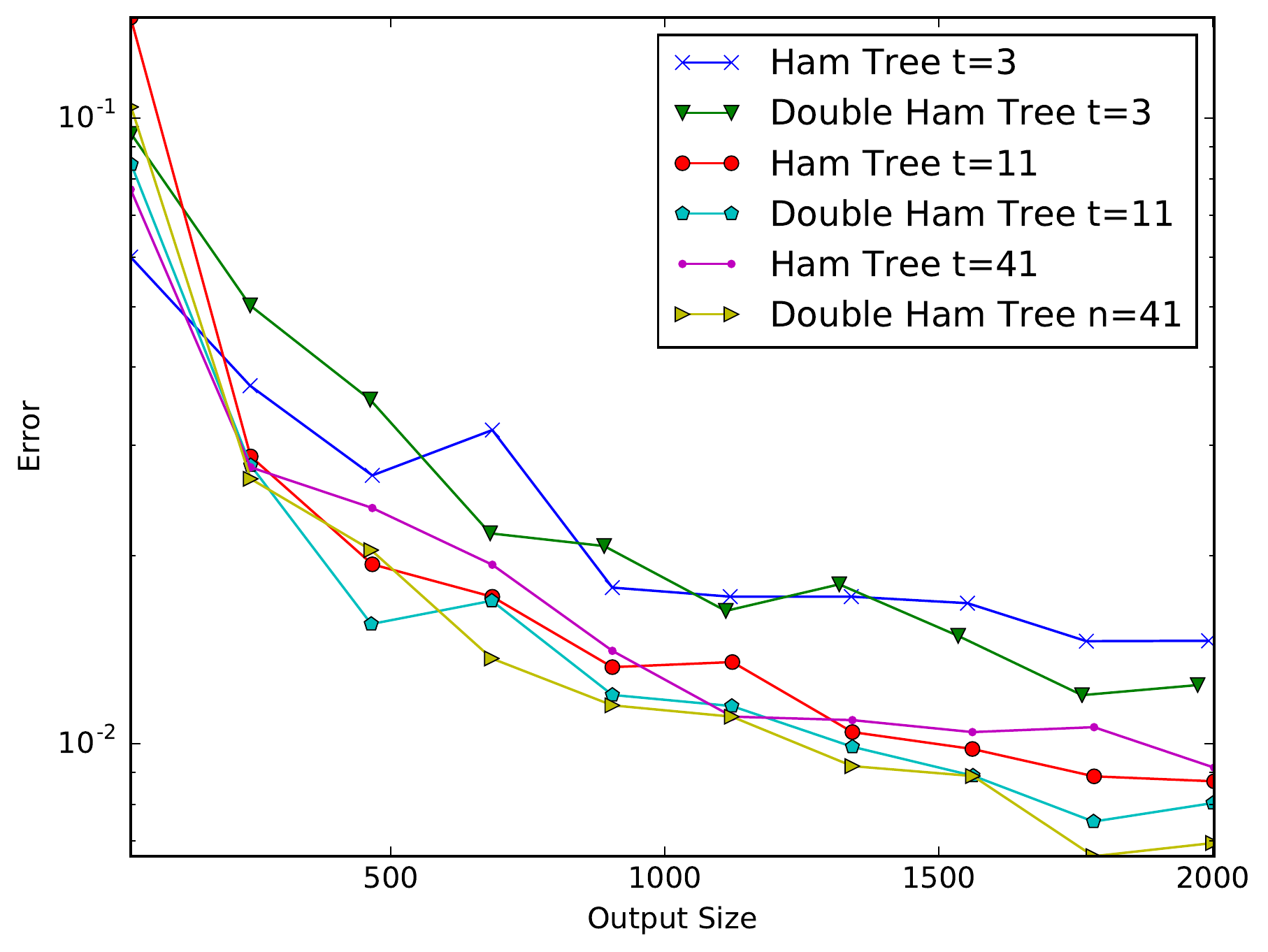}
	\hspace{3mm}
	\includegraphics[width=0.45\linewidth]{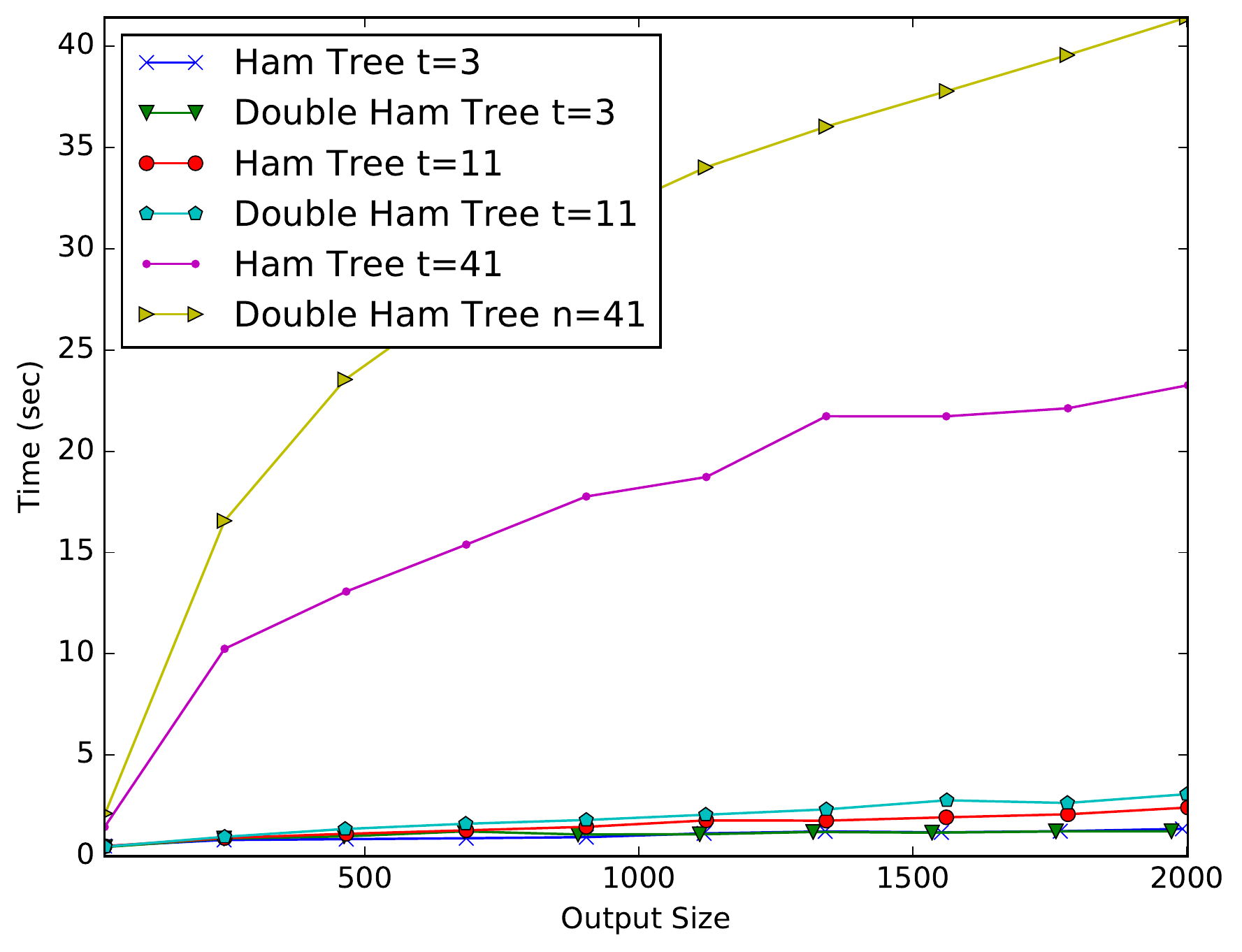}
	
	\vspace{-.18in}
	\caption{Output size vs. Time and Error for various $t$ with \Akd and \Akdd.}
	\label{fig:hamoutput}
\end{figure}

\section{Experiments on $\eps$-Samples and Applications}
\label{sec:exp}
In this section we explore the efficacy of our $\eps$-sample algorithms based on partitions.  
We use as $X$ the Chicago crime data~\cite{Chicago} with roughly 6.5 million data points.

A key step of the analysis is measuring the accuracy of the $\eps$-sample.  That is for a sampled $S$ we measure 
$\error(X,Y) = \max_{h \in \c{H}_d} | \frac{|Y \cap h|}{|Y|} - \frac{|X \cap h|}{|X|} |$, which unfortunately requires $|X|^{d+1}$ time to simply enumerate, which would be infeasible for large $X$.  
Instead we use techniques~\cite{SSSS,ABM06,Scanning} which provide guaranteed approximation of this function, designed with spatial anomaly detection in mind.  We have set the parameters large enough so the noise in computing \textsf{Error} is insignificant compared to the quantities we are evaluating.  
We evaluated the accuracy and efficiency of computing $\eps$-samples with $8$ different methods.  
\begin{itemize} 
\item There are 3 algorithms based on sampling one element per cell from Matousek's partition algorithm with polygonal cells, using tests created by lines \AMatL, points \AMatP, or the dual approach \AMatD.  
\item We consider 2 algorithms based on Chan's partition algorithm \AChan and \AChanS.  Each uses polygonal cells and the dual approach for the test set since this specific type of test set allowed for an optimization in the reweighting step.  The \AChan variant includes subsampling among cells for purpose of reweighting, while the \AChanS simply uses all of these cells and does not require the sampling step which in practice was inefficient.  
\item Then \Akd draws samples from the cells of the Willards partitioning; these are also cells of a partition, but with worse theoretical size-accuracy bounds.  
\item Finally we consider two baselines: random sampling, \ARS, and another approach \AEASE~\cite{ABM06} which is a greedy, but slow algorithm which has similar worst case guarantees to random sampling, but achieves better error in practice.  
\end{itemize}

\begin{figure}[t]
	\includegraphics[width=0.46\linewidth]{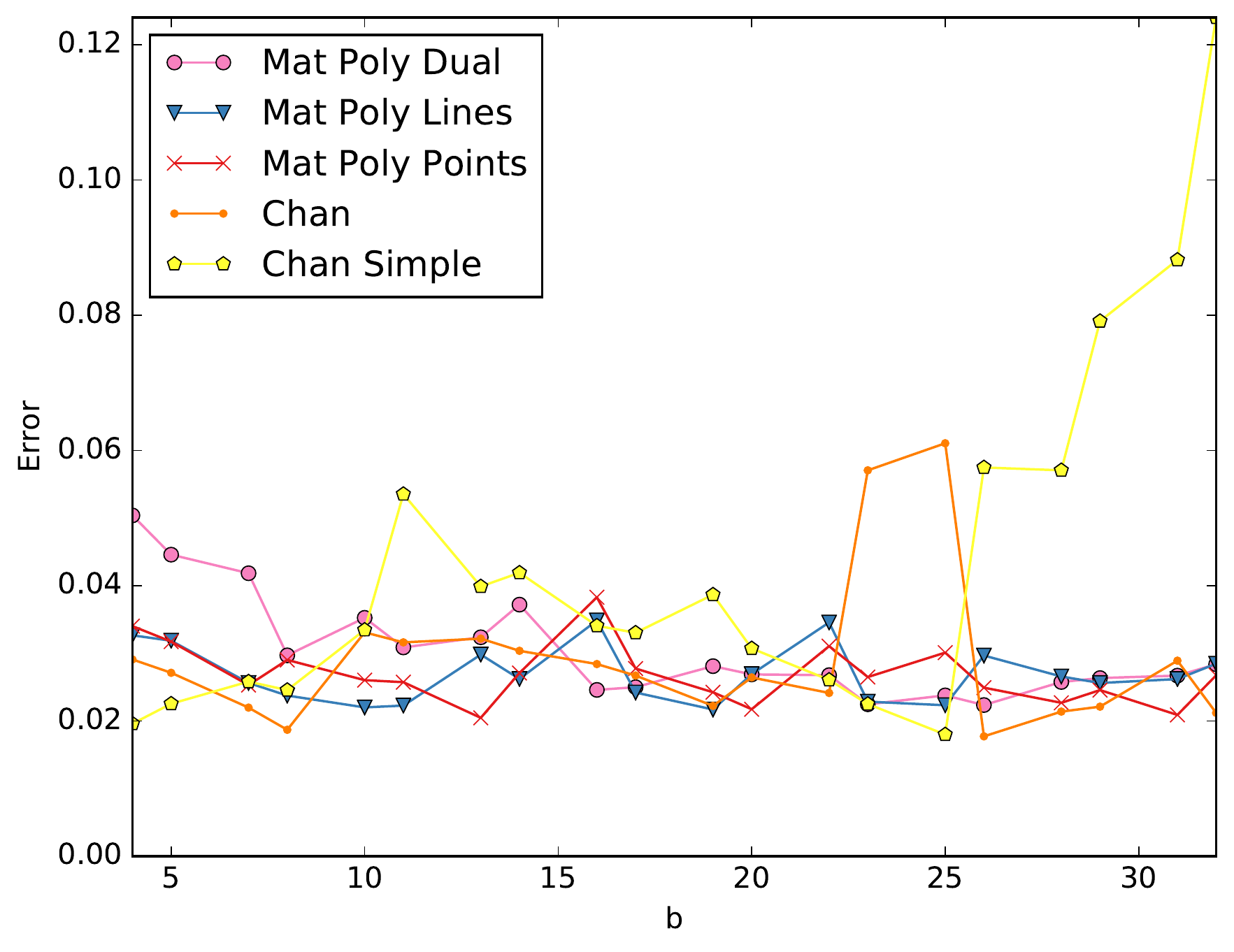}
	\hspace{3mm}
	\includegraphics[width=0.45\linewidth]{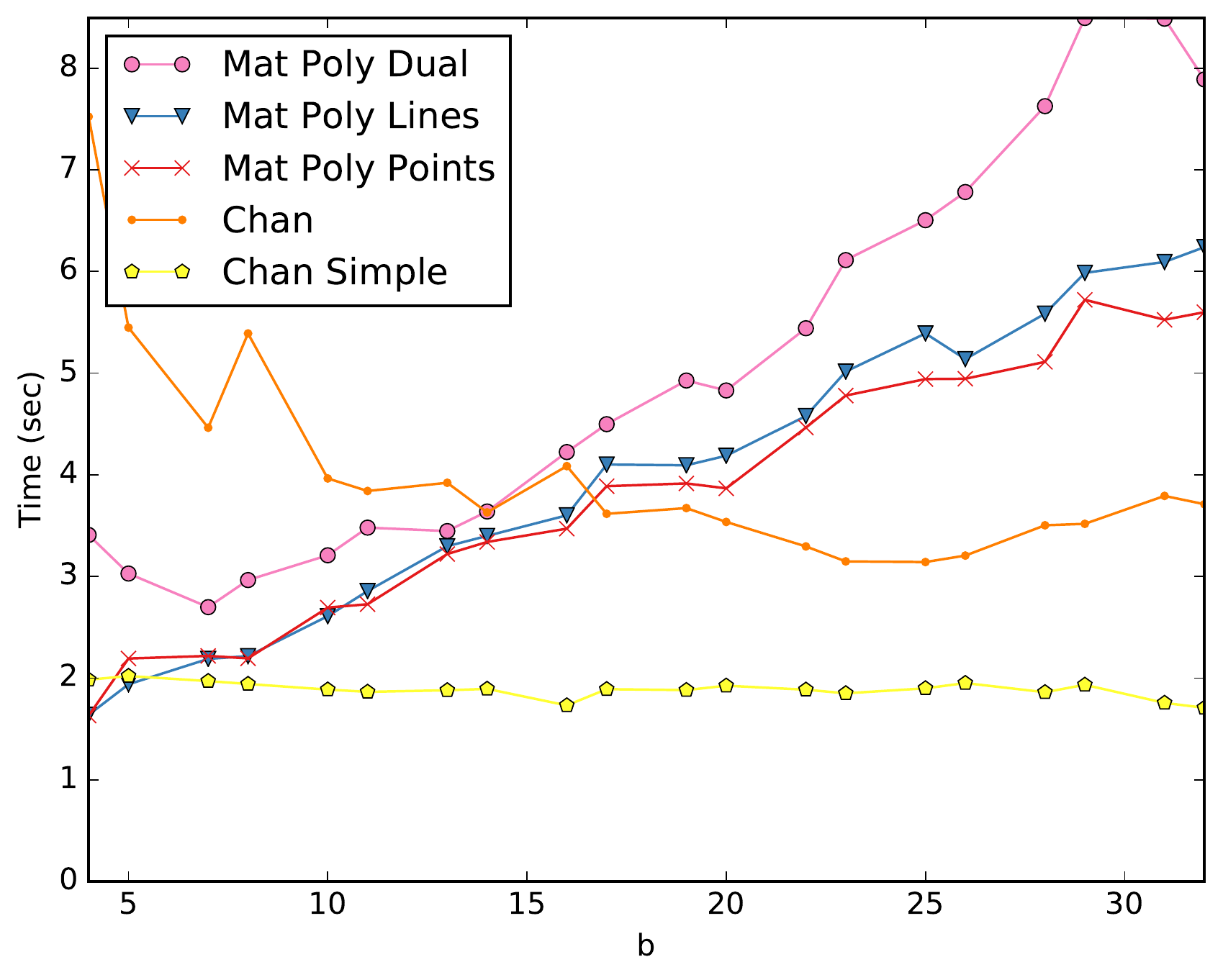}
	
	\vspace{-.18in}
	\caption{The \branch $b$ vs. Time and Error using the default parameters.}
	\label{fig:b_small}
\end{figure}

In testing these algorithms we can control three parameters: 
the \branch $b$, 
the \inS $n$ (default $n=100{,}000$ sampled from the crime data set), and 
the \out $k$ (default $k=1{,}000$).   We do not create a sample before creating the partition as analyzed in Theorem \ref{thm:halfplanesample}; we just create the partition on the $n$ points, then sample a point from each cell for the $\eps$-sample.  
The branching factor only effects \pmat (default $b=16$) and \pchan (default $b=22$) and is constant for the execution of the algorithm. 
\begin{figure}[t]

    \vspace{-3mm}
	\includegraphics[width=0.46\linewidth]{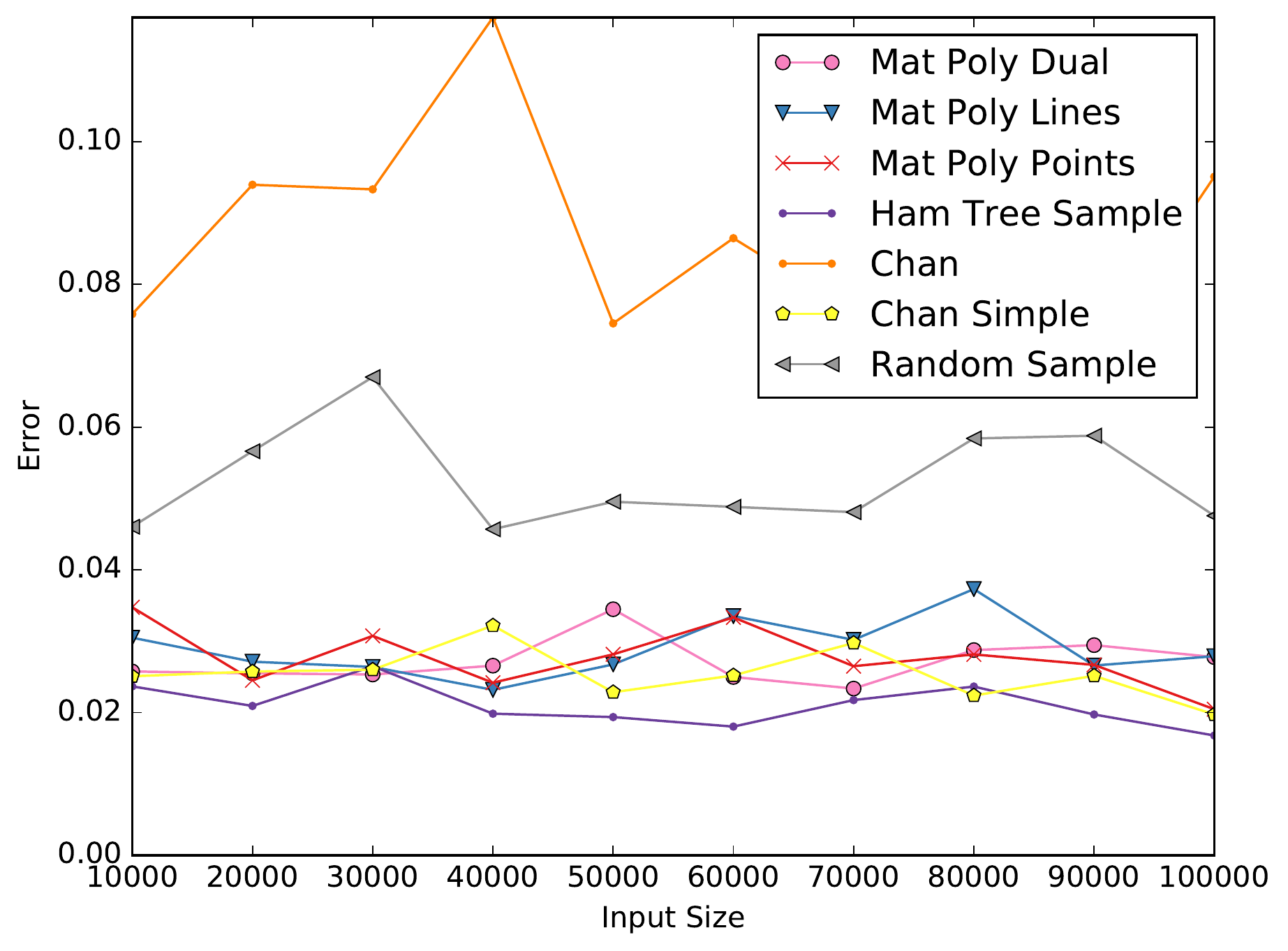}
	\hspace{3mm}
	\includegraphics[width=0.45\linewidth]{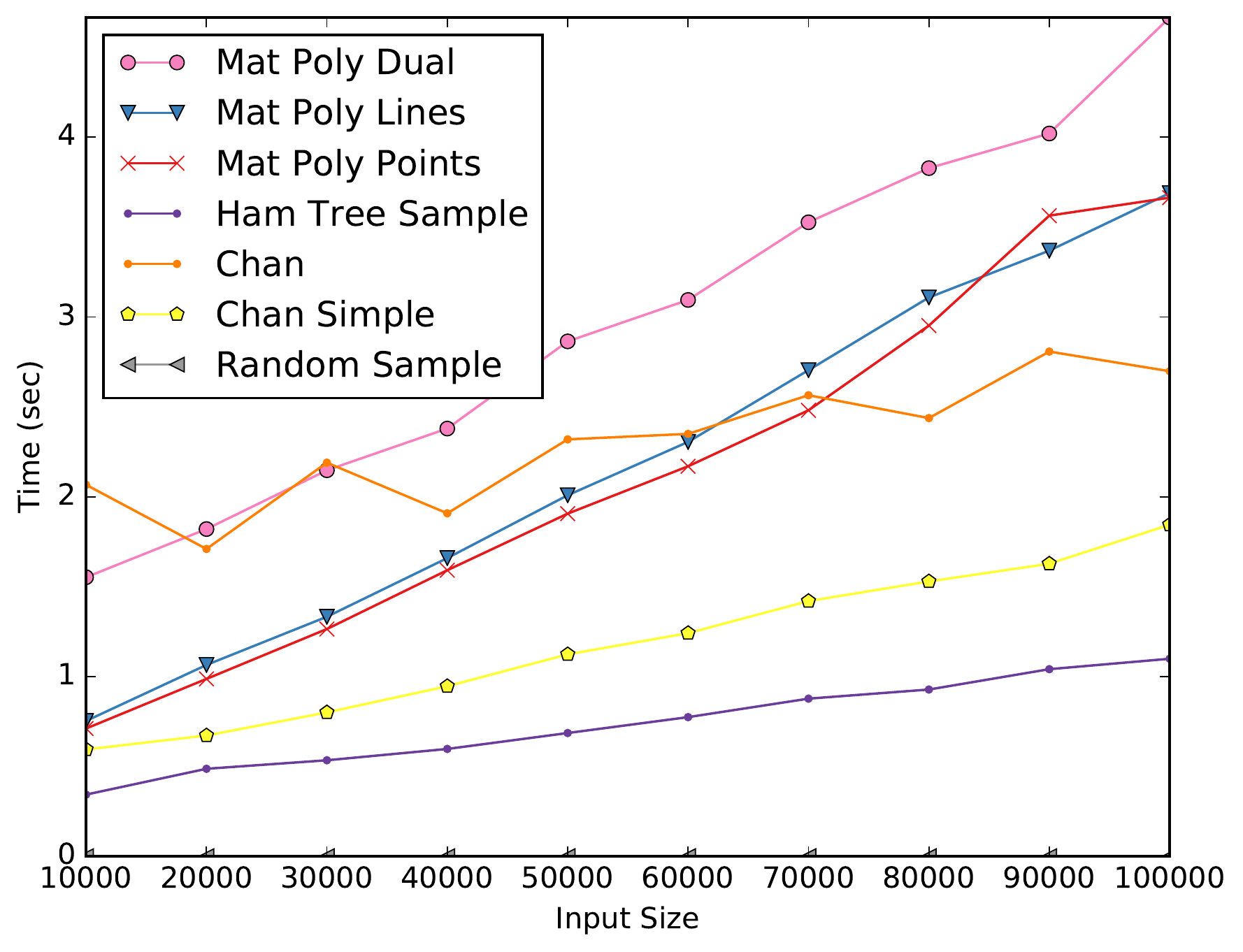}

	\vspace{-.18in}
	\caption{Input size vs. Time and Error using the default parameters.}
	\label{fig:is_small}
\end{figure}

\begin{figure}[t]
	\includegraphics[width=0.46\linewidth]{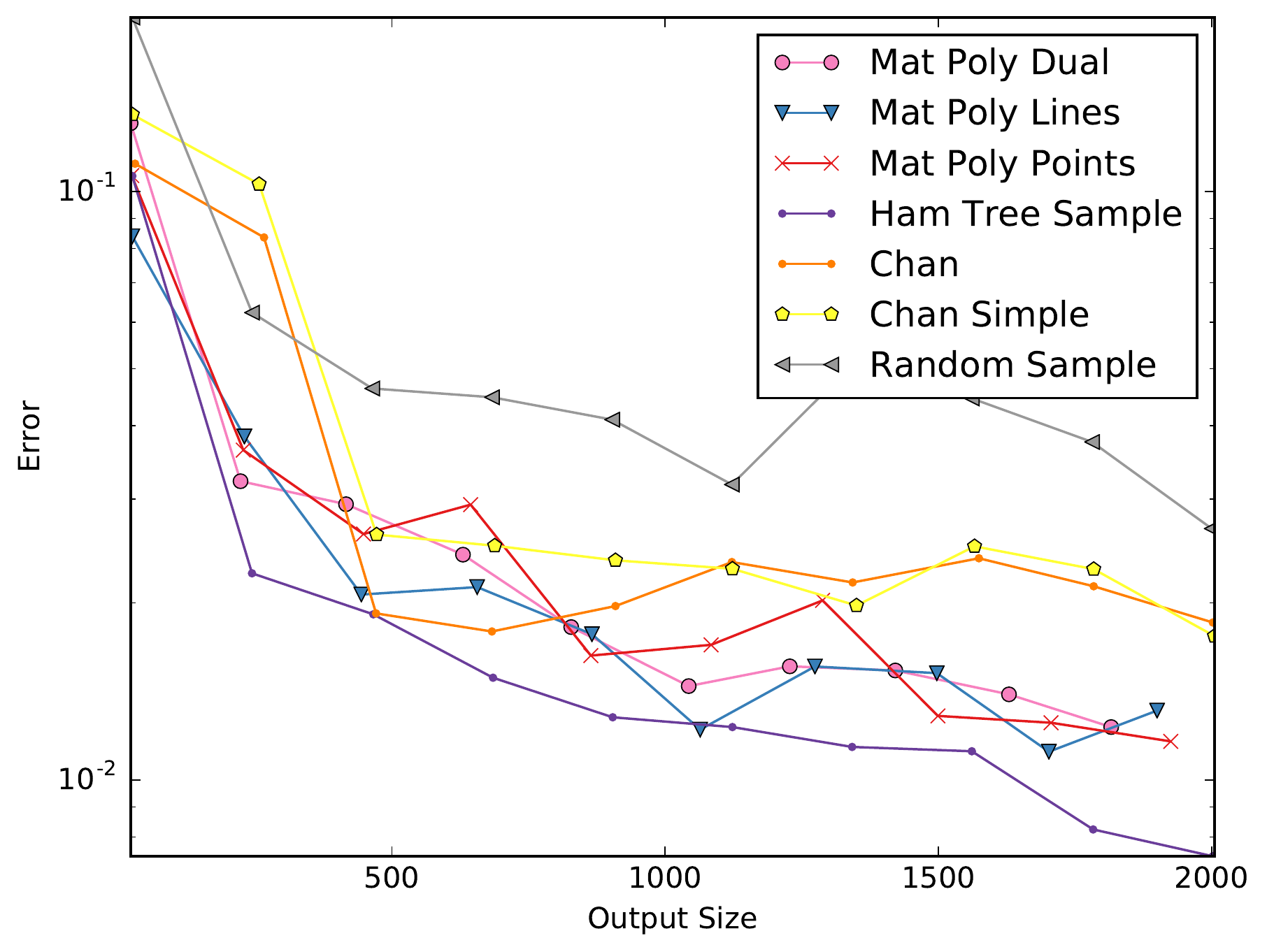}
	\hspace{3mm}
	\includegraphics[width=0.45\linewidth]{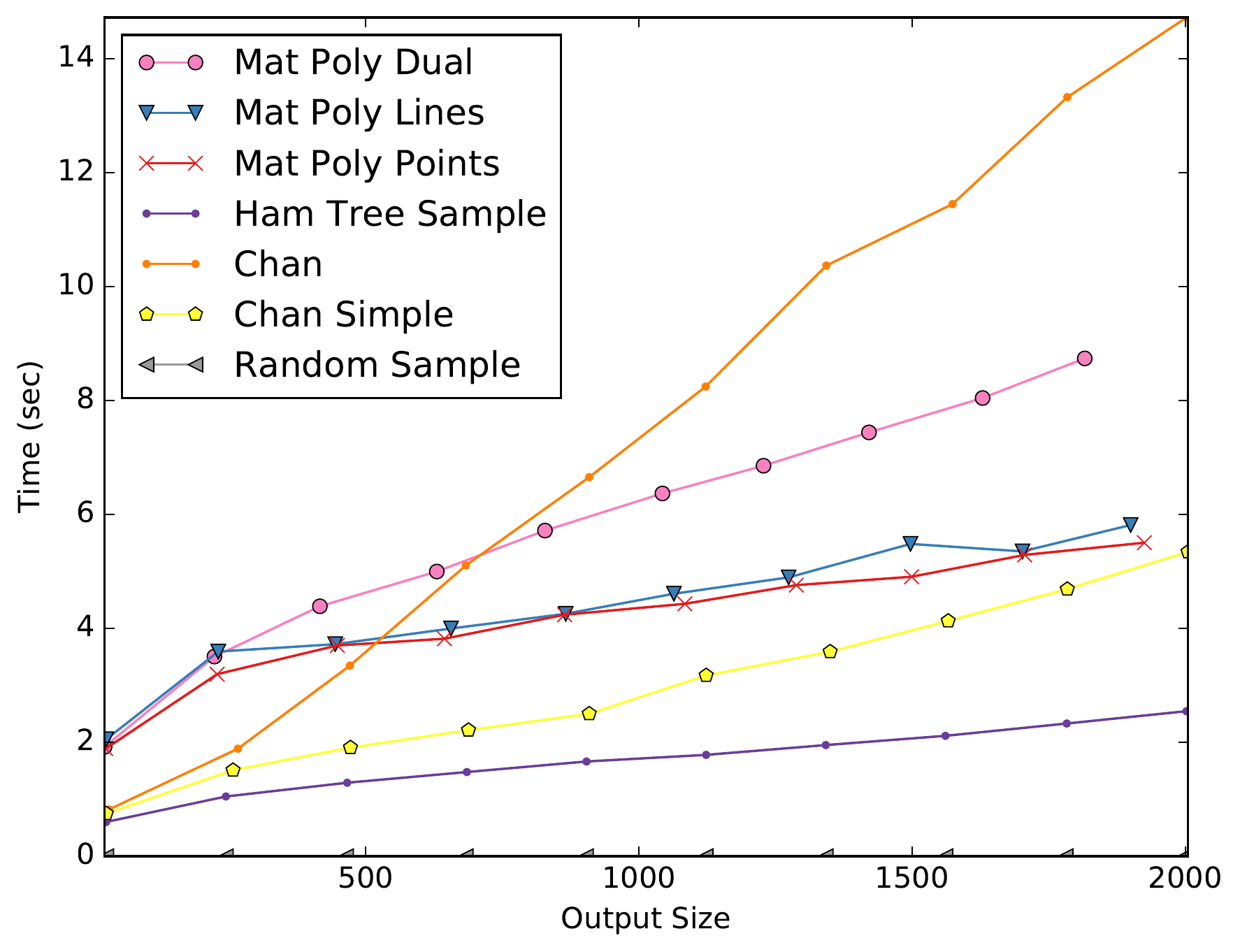}

	\vspace{-.18in}
	\caption{Output size vs. Time and Error using the default parameters. \vspace{-.1in}}
	\label{fig:os_small}
\end{figure}

\subparagraph*{Sample Evaluation Results.}
We do not plot \AEASE since it was quite slow as a function of the \out.  For $k=51$ it required $360$ seconds which was already more than a factor $100\times$ slower than any other algorithm, and became nearly intractable for $k > 100$.  We do note however that its measured \error on small $k$ is competitive with the best of our partitioning based methods.

Figure \ref{fig:b_small} shows how \branch $b$ affects the time and error. 
Matou\v{s}ek-based algorithms seem to gradually decrease in \error, but the trend is very small.  For \AChanS, the \error encounters a phase shift at around $b=25$, where the error suddenly becomes significantly worse for larger $b$, probably as an effect of the data set size.  
The timing is fairly unaffected by $b$ for Chan-based algorithms, but increases noticeably and linearly for the Matou\v{s}ek based algorithms.  We conclude that $b=22$ is a good choice for Chan-based algorithms and $b=16$ is a good choice for Matou\v{s}ek based algorithms.

Figure \ref{fig:is_small} shows the \inS relationship to time and \error.  As prescribed by the theory, \inS has no noticeable effect on \error.  Moreover, also as expected \textbf{the runtime of all algorithms scale linearly with \inS}.

Figure \ref{fig:os_small} plots the \out against the time and \error. 
As \out increases, as expected the error for all methods decreases, note that \error is plotted on a log-scale.  The \Akd and \AMatL achieve the smallest \error, with \Akd doing the best, and all proposed methods appear to improve upon the old default \ARS in terms of \error.  In particular with \out $k=1000$ both \AMatP and \Akd have $\error \approx 0.01$ while \ARS has $\error \approx 0.04$.  For the Matou\v{s}ek-based partitioning algorithms, the choice of test set does not have much effect on \error, and perform slightly worse than those based on Chan's partitioning.

Moreover, as \out increases the observed run time of all algorithms increases at most linearly.  In some cases (e.g., \Akd and \AMatL) the increase is sublinear as these are hierarchical methods, and the largest cost is incurred at the top of the hierarchy.  
Here as in other plots, we observe that \ARS is absurdly faster than any other approach.   However, even for \out $k=1000$, our methods \Akd, \AChanS, and \AMatP take only about $1$, $2.5$, and $4$ seconds, respectively.

\subparagraph*{Spatial Anomaly Detection Evaluation.}
As a concrete demonstration of the usefulness of efficient $\eps$-samples in practice, we apply our new algorithms to a framework for approximately detecting spatial anomalies -- maximizing the spatial scan statistic~\cite{Kul97}.  Specifically each point is endowed with two measures ($b(x)$ the baseline quantity like population and $m(x)$ the measured quantity like disease instance), and let $m(h)$ and $b(h)$ be the fraction of all measured and baseline counts within range $h \in \c{H}_d$, respectively.  The main computational problem of exact scan statistics is to find   
$h^* = \mathrm{arg}\max_{h \in \c{H}_d} \Phi(h)$ 
where for simplicity we use $\Phi(h) = \phi(m(h),b(h)) = |m(h)-b(h)|$.  

\begin{wrapfigure}{r}{0.49\linewidth}
	
	\vspace{-.15in}
	\includegraphics[width=\linewidth]{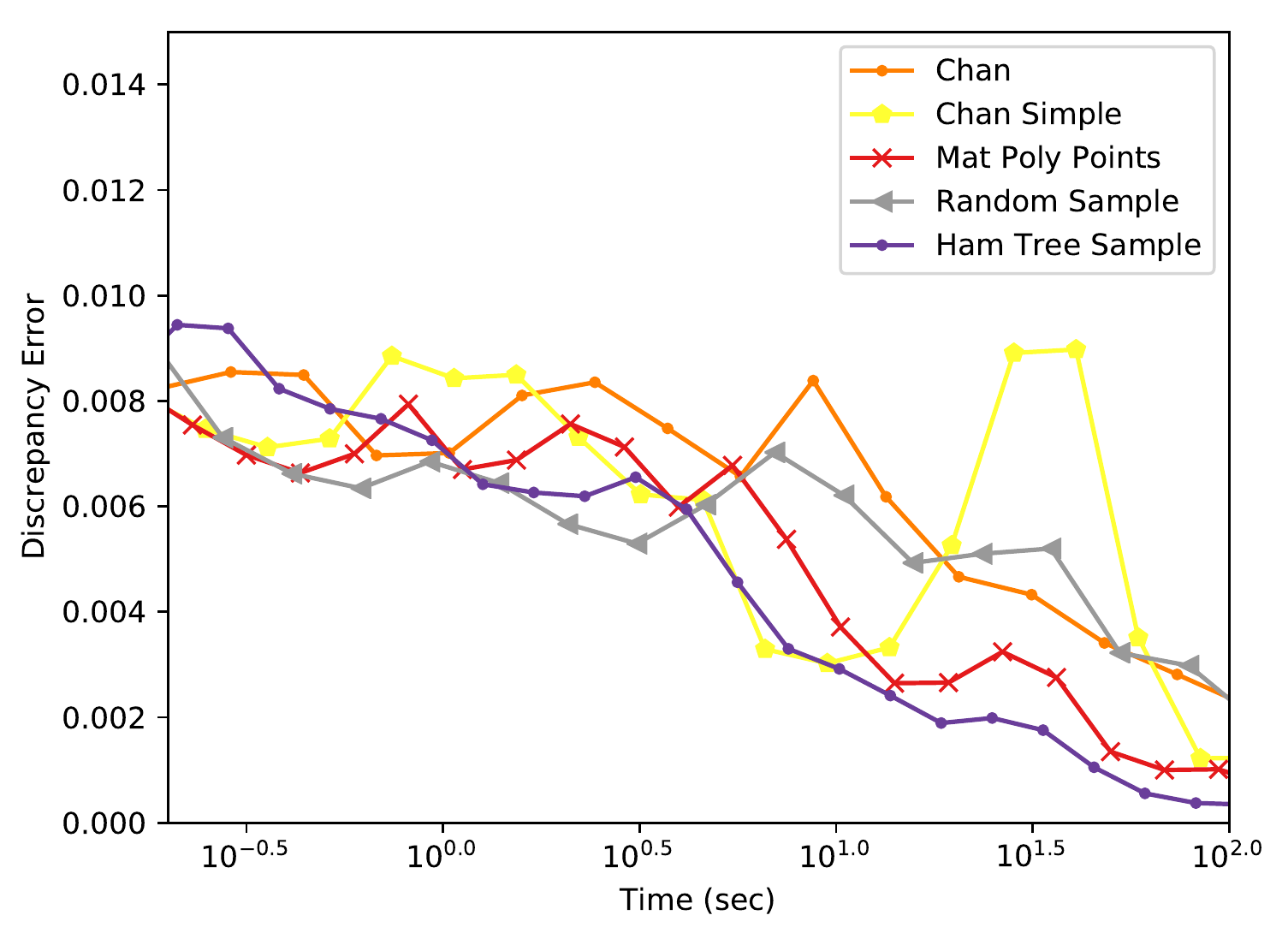}
	
	\vspace{-.16in}
	\caption{\label{fig:disc} Smooth \textsf{Discrepancy Error} vs. time.}
	\vspace{-.20in}
\end{wrapfigure}

Approximate scan statistics~\cite{SSSS,Scanning} depend on creating two samples an $\eps$-net which approximates the density of the regions and an $\eps$-samples which approximates the density of points. Together this allows the algorithm to find a 
$\hat h$ where $|\Phi(\hat h) - \Phi(h^*)| \leq \eps$; and this is still statistically powerful~\cite{SSSS}.  
In particular, we consider an algorithm for $\hat h$ which runs in time $O(n + \frac{1}{\eps} k \log \frac{1}{\eps} + T(n,k))$, where $k$ is the $\eps$-sample size and $T(n,k)$ its construction time. We fix $\eps$ to be approximately $.0025$ which corresponds approximately to an $\eps$-net of size $400$. and vary only $k$. `
We find approximate anomalies on the crime data set with a particular $h' \in \c{H}_d$ chosen and points chosen, so that $\Phi(h')$ will be anomalously large. Namely we plant a region containing $.02$ fraction of the points, where in that region points are in the measured set with probability of $.7$ and baseline set of $.3$ and outside with probability $.5$ and $.5$ respectively. In Figure \ref{fig:disc} we plot $\textsf{Discrepancy Error} = |\Phi(\hat h) - \Phi(h')|$ as a function of the overall runtime of the algorithms. Note that $\Phi(h^*) \ge \Phi(h')$, so it is possible to find a $\Phi(\hat h) \ge \Phi(h')$, but $\Phi(h')$ serves as a useful proxy. We find that \Akd generally outperforms \ARS; for instance for $0.003$ error, \Akd takes $10$ seconds to \ARS's $50$ seconds.  \AMatP also usually performs better than \Akd, while \AChan and \AChanS perform comparably to random sampling, albeit with high variance, even though their sampling procedure is hundreds of times slower.

\subparagraph*{Conclusion.}
Overall we recommend \Akd for computing $\eps$-samples if moderate computing beyond random sampling can be tolerated.  This method significantly reduces the size and error versus random sampling, and is not difficult to implement. If post-processing is not extensive, \ARS is still a simple reasonable choice in many settings. 

One thing to consider when implementing these algorithms is whether the extra complexity of a method such as \AChan is necessary. \AChan creates a tree structure that is very similar to \Akd or \Akdd, but differs in that it switches between partitioning based on points, \Akd, and partitioning based on lines, cuttings. Maintaining the global information for the cuttings leads to the much higher overhead of this method, but is only really necessary to avoid the possible construction of bad partitions. In most situations \Akd will perform just as well or better, since the data set will not be adversarial. For instance if the point set is uniformly distributed then even a $kd$-tree constructed partitioning will give an optimal $z=\frac{1}{2}$ \cite{Mat94}. In applications where consistently small sample sizes are extremely important and data can be manipulated by a 3rd party then guarantees become necessary.

\newpage
\bibliography{discrepancy}

\end{document}